\documentclass[11pt]{article}
\usepackage{subfig,amsmath,amssymb,amsthm,graphicx,tikz,bbm,bm,geometry,setspace, adjustbox, lscape, rotating, multirow, threeparttable, hyperref, siunitx}
\usepackage[authoryear,longnamesfirst]{natbib}
\newtheorem{lemma}{Lemma}
\newtheorem{corollary}{Corollary}
\newtheorem{theorem}{Theorem}
\newtheorem{condition}{Condition}
\newtheorem{definition}{Definition}
\newtheorem{assumption}{Assumption}
\newtheorem{remark}{Remark}
\DeclareMathOperator*{\argmax}{arg\,max}
\begin{document}
\title{Identification and Inference of Network Formation Games with Misclassified Links\thanks{First version: April, 2017. We would like to thank Elie Tamer, anonymous referees, Victor Aguirregabiria, Federico Bugni, Áureo de Paula, Bryan Graham, Shakeeb Khan, Michael Leung, Zhongjian Lin, Matt Masten, Arnaud Maurel, Konrad Menzel, Francesca Molinari, Kenichi Nagasawa, Pedro Souza, Roland Rathelot, Eric Renault, Camilla Roncoroni, Adam Rosen, and seminar participants at Bristol Econometric Study Group, Warwick Summer Workshop, 2018 European Meeting of the Econometric Society at Naples, and 2019 Asian Meeting of the Econometric Society at Xiamen for very helpful comments. The usual disclaimer applies.}}
\author{Luis E. Candelaria\thanks{Department of Economics, University of Warwick, Social Sciences Building, Coventry, CV4 7AL, U.K.; Email: L.Candelaria@warwick.ac.uk } \and Takuya Ura\thanks{Department of Economics, University of California, Davis, One Shields Avenue, Davis, CA 95616-5270; Email: takura@ucdavis.edu}}
\date{\today}
\maketitle
\begin{abstract}
This paper considers a network formation model when links are potentially misclassified. We focus on a game-theoretical model of strategic network formation with incomplete information, in which the linking decisions depend on agents' exogenous attributes and endogenous network characteristics. In the presence of link misclassification, we derive moment conditions that characterize the  identified set for the preference parameters associated with homophily and network externalities. Based on the moment equality conditions, we provide an inference method that is asymptotically valid when a single network of many agents is observed. Finally, we apply our misclassification-robust method to study the preference parameters of a lending network in rural villages in southern India. 
\begin{description}
\item Keywords: Misclassification, Network formation models, Strategic interactions, Incomplete information 
\item JEL Codes: C13, C31
\end{description}
\end{abstract}

\clearpage
\section{Introduction}\label{sec1}
Researchers across different disciplines have documented that measurement error of links is a pervasive problem in network data (e.g., \citealt{holland/leinhardt:1973}, \citealt{moffitt2001policy}, \citealt{kossinets:2006}, \citealt{ammermueller/pischke:2009}, \citealt{wangetal:2012}, \citealt{angrist:2014},  \citealt{depaula:2017}, \citealt{advani/malde:2018}). 
Although strategic network formation models provide essential information for learning about the creation of linking connections and peer effects when the network of interaction is endogenous, to the best of our knowledge, there has been no work addressing the effects of misclassifying links in strategic network formation models.
In this paper, we consider identification and inference in a game-theoretical model of strategic network formation with potentially misclassified links. 

We focus on a simultaneous game with incomplete information in which agents decide to form connections to maximize their expected utility (cf. \citealt{leung2015two}, and \citealt{ridder/sheng:2020}). 
The agents' decisions are interdependent since the utility attached to creating a link depends on the agents' observed attributes and network characteristics through link externalities (such as reciprocity, in-degree, and transitivity statistics). 
The misclassification problem will affect the link formation decisions in two different ways. 
First, the binary outcome variable representing an agent's optimal linking decision is misclassified. 
Second, the link misclassification problem prevents us from directly identifying the belief system that an agent uses to predict others' linking decisions.
In this sense, the misclassification problem occurs on the left- and right-hand sides of the equation describing the optimal linking decisions (as shown in Lemma \ref{lemma1}).

We propose a novel approach for analyzing network formation models, which is robust to link misclassification. 
Specifically, in a setup that allows for the links to be potentially misclassified, we characterize the identified set for the structural parameters, including the preference parameters concerning homophily and network externalities.
A notable innovation in our approach is that we derive the relationship between the choice probabilities of observed network connections and the belief system (as shown in Lemma \ref{belief_lemma}). This result is crucial in allowing us to control for the endogeneity of the equilibrium beliefs and to reduce the model to a single-agent decision model in the presence of misclassification.

We also propose an inference method that is asymptotically valid when one network with a large number of agents is observed. 
Our proposed confidence interval is computationally feasible and controls the size even when the parameters are partially identified. 

In practice, network links might be misclassified due to several reasons. For instance, true links might be listed as missing due to incomplete surveying of the individuals in the network, top-censoring of the number of links that can be reported, or because individuals might be apprehensive about revealing certain connections, they might get wearied during the interview or forget about some connections. Meanwhile, non-existent links might be listed as present due to perception biases that may lead individuals to list incorrect connections or mistakes in the imputation of links by the researcher. In our setting, the links are allowed to be misclassified conditionally at random as positives or negatives. The formal statement is provided in Assumption \ref{Ass3}, where we also discuss the advantages and limitations of assuming a misclassification process that is conditionally at random. In Appendix \ref{sec_aprob}, we discuss a generalization of the misclassification process to a heterogeneous setting where the misclassification distribution is covariate-dependent.

In an empirical illustration, we apply our inference method to examine a lending network in Karnataka, an area in southern India (see \citealt{jackson/rodriguezbarraquer/tan:2012} and \citealt{banerjee/chandrasekhar/duflo/jackson:2013}). We study the performance of our method when different degrees of links misclassification are considered in the data, including both false positives and negatives. As a benchmark scenario, we use  the no link misclassification case. This corresponds to the analysis conducted by \cite{leung2015two}. 
Our results suggest that, even with misclassified links, the most important determinants driving the lending decisions of the individuals in the network are reciprocation, homophily on gender, and whether or not the individuals are relatives. 
Moreover, this analysis documents the importance of using an inference method robust to the misclassification of network links. For instance, in a scenario that controls for up to 50\% probability of link misclassification as negatives, the 95\%  confidence intervals obtained using our inference method can be up to $2.95$ times larger in length than the confidence intervals that ignore the misclassification problem.

Our methodology contributes to the growing econometric literature studying the formation of networks (see \citet{graham:2015}, \citet{chandrasekhar:2016}, and \citet{depaula:2017} for an overview). Within this literature, the studies by \citet{leung2015two, menzel2015strategic, miyauchi:2016, boucher/mourifie:2017, mele:2017, depaula/richards-shubik/tamer:2018, thirkettle:2019, christakis/fowler/imbens/kal:2020, ridder/sheng:2020, sheng:2020, badev:2021}, and \citet{gualdani:2021} have analyzed game-theoretic models of network formation. The papers most similar to ours are those by \citet{leung2015two} and  \citet{ridder/sheng:2020}, which analyze the strategic formation of a directed network with incomplete information and network externalities. 
Relative to these papers, our paper allows for the links to be potentially misclassified and examines the problems arising from that misclassification in identifying and estimating the network formation model.\footnote{A different strand from this literature has studied dyadic link formation models with unobserved node-specific heterogeneity  (see, e.g., \citealt{graham2015empirical,charbonneau:2017, toth:2017, jochmans:2018, dzemski:2019, candelaria:2020,gao:2020, ma/su/zhang:2020, zeleneev:2020, auerbach:2022}). These papers rule out the presence of link misclassification in their setup.} 

Within the existing literature of network formation models,  \citet{chandrasekhar/lewis:2016} and \citet{thirkettle:2019} have examined the effects of partially observed network data on recovering or bounding network statistics, such as the Katz-Bonacich centrality or clustering measures.
Their methodologies complement ours as they focus on sampled networks where the observed part of the network is assumed not to suffer from any measurement error. Moreover, their asymptotic framework assumes that multiple small networks are available to the researcher. 
In contrast, our methodology allows for all the links in the network to be potentially misclassified, e.g., the observed links could represent false positives with a fixed probability. In addition, our setting is designed for the case when only one large network is available to the researcher. 

This paper is also related to the literature of mismeasured discrete variables (e.g.,   \citet{hausman/abrevaya/schott-morton:1998,mahajan:2006,lewbel:2007,chen/hu/lewbel:2008,hu:2008,molinari:2008,hu/lin:2018}).\footnote{\citet{hu/lin:2018} analyzes a social interaction model where the decision variable is mismeasured, but the network of interaction is observed without error.}  
Our approach to misclassified links is based on \cite{molinari:2008}, who offers a general bounding strategy with misclassified discrete variables.  
Specifically, we can also incorporate ex-ante restrictions on the misclassification probabilities and bound the underlying parameter of interest.

There are several papers in the literature of social interactions that have examined the econometric challenges posed by the limited availability of network data. The settings considered in those studies include partially observed links (\citealp{blume/brock/durlauf/jayaram:2015, kline:2015, lewbel/qu/tang:2019, griffith:2019, zhang2020}), sampled network data (\citealp{liu:2013, ata/belloni/candogan:2018}), aggregated network data (\citealp{alidaee/auerbach/leung:2020, boucher/houndetoungan:2020, breza/chandrasekhar/mcCormick/pan:2020}), and completely unobserved network data (\citealp{rose:2015, manresa:2016, depaula/rasul/souza:2019b}). However, their methods differ substantially from the ones proposed here since those papers have different objects of interest. In particular, they focus on estimating peer effects, network externalities, or network statistics, instead of investigating the effects that limited network data have on the preference parameters that drive the formation of a network. In fact,  most of these papers take the network of interactions to be exogenously determined. In contrast, our paper studies the effects of link misclassification on the preference parameters characterizing a strategic network formation model.

The remainder of the paper is organized as follows. 
Section \ref{sec2} describes the network formation model as a game of incomplete information.
Section \ref{sec3} characterizes the identified set for the structural parameters.
Section \ref{sec4} introduces an inference method based on the representation of the identified set. 
Sections \ref{sec5} presents an empirical application using data on a lending network in rural villages in southern India. 
Section \ref{sec6} provides concluding remarks. 
The appendix collects additional results and all the proofs of the paper. 

\section{Network Formation Game with Misclassification}\label{sec2}

We extend the directed network formation model discussed in \cite{leung2015two} and  \cite{ridder/sheng:2020} to account for potentially misclassified links. 
In particular, we use a static game of incomplete information as a framework to model the formation of a directed network.
For simplicity, our approach follows \cite{leung2015two}.

Consider a network determined by a set of $n$ agents, which we denote by $\mathcal{N}_n=\{1,\ldots,n\}$.  
We assume that each pair of agents $(i,j)$ with $i,j \in \mathcal{N}_n$ is endowed with a vector of exogenous attributes $X_{ij} \in \mathbb{R}^d$ and an idiosyncratic shock $\varepsilon_{ij} \in \mathbb{R}$.
Let $X=\{X_{ij}: i,j \in \mathcal{N}_{n}\}\in \mathcal{X}^n$ be a profile of attributes that is common knowledge to all the agents in the network,  $\varepsilon_i=\{\varepsilon_{ij}:j\in\mathcal{N}_n\}$ be a profile of idiosyncratic shocks that is agent $i$'s private information, and $\varepsilon=\{\varepsilon_i:i\in\mathcal{N}_n\}$ collects all the profiles of idiosyncratic shocks.

The network is represented by an $n\times n$ adjacency matrix $G_n^{\ast}$, where the $ij$th element $G^{\ast}_{ij,n} = 1$ if agent $i$ forms a direct link to agent $j$ and $G^{\ast}_{ij,n}=0$ otherwise. 
We assume that the network is directed, i.e., $G^{\ast}_{ij,n}$ and $G^{\ast}_{ji,n}$ may be different. The diagonal elements are normalized to be equal to zero, i.e., $G^\ast_{ii,n}=0$. 
The researcher observes $G_{n}$, a proxy of the true underlying network $G_{n}^\ast$ with potentially misclassified links. 

Given the network $G_{n}^\ast$ and information $(X,\varepsilon_i)$, agent $i$ has utility
\begin{align*}
U_{i}(G_{i,n}^\ast, G_{-i,n}^\ast, X, \varepsilon_{i}) 
= \frac{1}{n} \sum_{j=1}^{n} G_{ij,n}^\ast
 \left[    
\left(
G_{ji,n}^\ast,\frac{1}{n}\sum_{k\ne i}G_{kj,n}^\ast,\frac{1}{n}\sum_{k\ne i, j}G_{ki,n}^\ast G_{kj,n}^\ast, X_{ij}' 
\right)\beta_0
+ \varepsilon_{ij}
 \right],
\end{align*} 
where $G^\ast_{i,n}=\{G^{\ast}_{ij,n}: j\in\mathcal{N}_n\}$, $G^\ast_{-i,n} = \{G^\ast_{j,n}: j\neq i\}$, and $\beta_0$ is an unknown finite dimensional vector in a parameter space $\mathcal{B}$.

Agent $i$'s marginal utility of forming the link $G_{ij,n}^\ast$ depends on a vector of network statistics, the profile of exogenous attributes, and the link-specific idiosyncratic component.\footnote{For simplicity, we consider three different kinds of factors in the vector of network statistics. It is straightforward to generalize our results to the complete specification in \cite{leung2015two}. A similar specification of the utility function has been used in \citet{mele:2017,thirkettle:2019}, and \citet{sheng:2020}.}
The first component in the vector of network statistics captures the utility obtained from a reciprocated link with agent $j$, $G^{\ast}_{ji,n}$. 
The second network statistic represents the in-degree of agent $j$, $\frac{1}{n}\sum_{k\ne i}G_{kj,n}^\ast$, which captures the utility obtained from connecting with agents of high centrality in the network.  
The last network statistic captures the utility of being connected to the same agents, $\frac{1}{n}\sum_{k\ne i, j}G_{ki,n}^\ast G_{kj,n}^\ast$. 
The profile of exogenous attributes captures the preferences for homophily on observed characteristics. Finally, $\varepsilon_{ij}$ is an unobserved link-specific component affecting agent $i$'s decision to link with agent $j$.

Let $\delta_{i,n}(X, \varepsilon_{i})$ denote a generic agent $i$'s pure strategy, which maps the information available to agent $i$, $(X,\varepsilon_{i})$, to an action in $\mathcal{G}^{n} = \{0,1\}^{n}$. Let $\sigma_{i,n}(g_{i,n}^\ast\mid X) =  Pr(\delta_{i,n}(X,\varepsilon_{i})=g_{i,n}^\ast\mid X)$ be the probability that agent $i$ chooses action $g_{i,n}^\ast \in \mathcal{G}^{n}$ given $X$ and $\sigma_{n}(X) = \{ \sigma_{i,n}(g_{i,n}^\ast\mid X), i \in \mathcal{N}_{n}, g_{i,n}^\ast\in \mathcal{G}^{n} \}$. We call $\sigma_{n}(X)$ a belief profile. Given a belief profile $\sigma_n$ and the information $(X,\varepsilon_i)$, agent $i$ chooses $g_{i,n}^\ast$ from $\mathcal{G}^{n}$ to maximize the expected utility $U_{i}(g_{i,n}^\ast, \delta_{-i,n}(X,\varepsilon_{-i}), X, \varepsilon_{i})$ given $(X,\varepsilon_{i},\sigma_{n})$. 

In an $n$-player game, a Bayesian Nash equilibrium $\sigma_{n}(X)$ is a belief profile that satisfies
\begin{equation*}
   \sigma_{i,n}(g_{i,n}^\ast\mid X) = Pr(\delta_{i,n}(X,\varepsilon_{i})=g_{i,n}^\ast\mid X, \sigma_{n})
\end{equation*}
for all $X \in \mathcal{X}^n$,  $g_{i,n}^\ast \in \mathcal{G}^n$, and $i \in \mathcal{N}_{n}$, where
\begin{equation*}
\delta_{i,n}(X,\varepsilon_{i})= \argmax_{g_{i,n}^\ast \in \mathcal{G}^n} E\left[U_{i}(g_{i,n}^\ast, \delta_{-i,n}(X,\varepsilon_{-i}), X, \varepsilon_{i}) \mid X, \varepsilon_{i} ,\sigma_n\right].   
\end{equation*}

We impose the following assumption on the observed attributes and idiosyncratic shocks, which also has been used by \cite{leung2015two} and \cite{ridder/sheng:2020}.
\begin{assumption}
\label{Ass1}
\begin{enumerate}
\item[(i)] For any $A_{1}, A_{2} \subset \mathcal{N}_{n}$ disjoint, $\{X_{ij}: i,j \in A_{1}\}$ and $\{X_{kl}: k,l \in A_{2}\}$ are independent. 
\item[(ii)] $\{\varepsilon_{ij}: i,j\in\mathcal{N}_{n}\}$ are identically distributed with the standard normal distribution, and $\{\varepsilon_{i}: i \in \mathcal{N}_{n} \}$ are independent. (We use $\Phi$ for the standard normal cdf and $\phi$ for the pdf.) 
\item[(iii)] $\varepsilon$ and $X$ are independent.
\item[(iv)] Attributes $\{X_{ij}:i,j\in\mathcal{N}_{n}\}$ are identically distributed with finite support $\mathcal{X}=\left\{ x_1,\cdots, x_J \right\}$, and $Pr\left( x_j \right)>0$ for all $x_j \in \mathcal{X}$.
\end{enumerate}
\end{assumption}

Condition (i) allows for correlation across the pairs of attributes $X_{ij}$ and $X_{kl}$ if they have a common index (i.e., $i=k$). 
As a consequence, the attributes across all the dyads formed by one agent may be dependent.
Condition (ii) assumes that the idiosyncratic shocks are identically distributed with known standard normal distribution. 
This represents a regularity condition as the results can be easily adjusted to the case that $\varepsilon_{ij}$ has an absolutely continuous distribution $F_{\varepsilon}(\cdot; \theta_{\varepsilon})$  that is known up to a finite-dimensional parameter  $\theta_{\varepsilon}$. In Appendix \ref{sec_a2}, we relax this assumption to a setting where the distribution of $\varepsilon_{ij}$ is not parametrically restricted and characterize the identified set in a semiparametric framework.
Notice that this condition also implies that the components of $\varepsilon_{i}$ may be arbitrarily correlated.
Condition (iii) rules out the possibility of agents learning about others' private information from the observed profile of attributes that is common knowledge. 
Condition (iv) assumes that $X_{ij}$ is a discrete random vector with finite support and will be used to prove uniform convergence of a nonparametric first-stage estimator. 
Requiring that $X_{ij}$ is discretely distributed is not necessary for our inference method and can be relaxed. 
See \citet[Appendix B]{leung2015two} for further details when $X_{ij}$ is continuously distributed.

We focus on a symmetric equilibrium for our inference method \citep[cf.][]{leung2015two}. This approach is suitable when the identities of the individuals in the network are irrelevant. An equilibrium profile $\sigma_{n}$ is symmetric if $\sigma_{i,n}(g_{i,n}^\ast \mid X)=\sigma_{\pi(i),n}(\pi(g_{\pi(i),n}^\ast)\mid \pi(X))$ for any $i \in \mathcal{N}_{n}$, $g_{i,n}^\ast \in \mathcal{G}^n$, and any permutation function $\pi \in \Pi$.\footnote{The permutation function $\pi \in \Pi$ is defined as in \citet[][Page 184]{leung2015two}.} Using Assumption \ref{Ass1}, \citet[][Theorem 1]{leung2015two} has shown the existence of a symmetric equilibrium (cf. \citet[][Proposition 1]{ridder/sheng:2020}). We take that result as given and summarize it in the next assumption.

\begin{assumption}
\label{Ass2}
For any $n$, the agents play a symmetric equilibrium $\sigma_n$, i.e., there exists $\{\delta_{i,n}:i\in \mathcal{N}_{n}\}$ such that for any $i \in \mathcal{N}_{n}$ the following holds: 
\begin{enumerate}
   \item[(i)] $G_{i,n}^\ast = \delta_{i,n}(X,\varepsilon_{i})$.
   \item[(ii)] $\sigma_{i,n}(g_{i,n}^\ast \mid X)=Pr(\delta_{i,n}(X,\varepsilon_{i})=g_{i,n}^\ast\mid X, \sigma_{n})$.
   \item[(iii)] $\delta_{i,n}(X,\varepsilon_{i})= \argmax_{g_{i,n}^\ast \in \mathcal{G}^n} E\left[U_{i}(g_{i,n}^\ast, \delta_{-i,n}(X,\varepsilon_{-i}), X, \varepsilon_{i}) \mid X, \varepsilon_{i} ,\sigma_n\right]$.
   \item[(iv)] $\sigma_n$ is symmetric.    
\end{enumerate}
\end{assumption}

Under our characterization of the equilibrium, the true underlying network is rationalized by a symmetric equilibrium in a network game with $n$ agents. Implicitly in Assumption 2, it is required that if there are multiple equilibria consistent with the model, the agents coordinate on a symmetric equilibrium using an equilibrium selection mechanism. Finally, since the inference method that we introduced is based on a large-network asymptotics that is conditional on the realized equilibrium,  uniqueness of the equilibrium is unnecessary (cf \citealt{leung2015two}, and \citealt{ridder/sheng:2020}). 

The next lemma characterizes the optimal decision rule for the formation of each link in the network. 
\begin{lemma}(Theorem 1, \citep{leung2015two})
\label{lemma1}
Under Assumption \ref{Ass1} and \ref{Ass2} , $G_{ij,n}^\ast=1 \left\{(Z_{ij,n}^\ast)'\beta_0+ \varepsilon_{ij}\geq 0\right\}$, where 
$$
 \gamma_{ij,n}^\ast =  E \left[ \left(G_{ji,n}^\ast,\frac{1}{n}\sum_{k\ne i}G_{kj,n}^\ast,\frac{1}{n}\sum_{k\ne i,j}G_{ki,n}^\ast G_{kj,n}^\ast\right)' \mid X, \sigma_n  \right]
$$
and
$$
Z_{ij,n}^\ast=\left(\begin{array}{c}\gamma_{ij,n}^\ast\\X_{ij}\end{array}\right).
$$
\end{lemma}
A direct implication of Lemma \ref{lemma1} is that each agent makes separate linking decisions for each of her potential links.
Given the misclassification problem, both the optimal action $G_{ij,n}^\ast$ and the equilibrium beliefs about the network statistics $\gamma_{ij,n}^\ast$ in the optimal decision rule will be misclassified.
In other words, the misclassification problem affects both left and right-hand side variables in the optimal decision rule.

We assume that the conditional distribution of the observed network $G_n$ is related to that of the true state of the network $G_n^\ast$ as follows.

\begin{assumption}
\label{Ass3}
There are two unknown real numbers $\rho_{0}, \rho_{1} \in [0,1)$ with $\rho_{0}+\rho_{1}<1$ such that the following two statements hold for every $n$ and every $i,j,k\in\mathcal{N}_n$.
(i) $G_{ki,n}$ and $G_{kj,n}$ are independent given $(G_{ki,n}^\ast, G_{kj,n}^\ast,X,\sigma_n)$. (ii) $Pr(G_{ij,n}\ne G_{ij,n}^\ast\mid G_{ij,n}^\ast,X,\sigma_n)=\rho_01\{G_{ij,n}^\ast=0\}+\rho_11\{G_{ij,n}^\ast=1\}$.
\end{assumption}

Assumption \ref{Ass3} states that given the information $X,\sigma_n$, the misclassification of each link is conditionally at random with unknown probabilities $\rho_0$ and $\rho_1$. In particular, Condition $(i)$ ensures that the misclassification of links can be treated as a conditionally independent process across dyads $\left\{ ki \right\}$ and $\left\{ kj \right\}$ for $i\neq j$.
This condition is imposed for simplicity and used exclusively to control for the nonlinear endogenous factor $\frac{1}{n}\sum_{k\ne i,j}G_{ki,n}^\ast G_{kj,n}^\ast$ in terms of its observed counterpart; Lemma \ref{belief_lemma} provides the exact statement. In other words, it plays no role in controlling for the reciprocity and in-degree statistics. Moreover, this condition becomes redundant when the unobserved links-specific components are assumed to be independently distributed, as it has been used in \citet{menzel2015strategic,chandrasekhar/lewis:2016, thirkettle:2019,ridder/sheng:2020} and \citet{sheng:2020}. Remark \ref{remark:independent_types} shows this result.

Condition $(ii)$ characterizes the misclassification probabilities, and it states that given the information $X,\sigma_n$, the links can be misclassified as positives with probability $\rho_0$ and as negatives with probability $\rho_1$. These misclassification probabilities are unknown and the degree of noise in the data cannot exceed 1, i.e., $\rho_0 + \rho_1 <1$. \citet{hausman/abrevaya/schott-morton:1998} have also used Condition $(ii)$, but in the setting of a binary choice model with misclassification of the dependent variable. A similar assumption has been used in \citet{molinari:2008} to bound the latent misclassification distribution. In this paper, we implement a bounding strategy that follows a similar logic to \citet{molinari:2008}. 

Assumption \ref{Ass3} provides a simple and intuitive characterization of the misclassification process, which is suitable to rationalize a setting where network links can be misclassified at random as negatives or positives. The misclassification probabilities can be asymmetrical to represent different sources of link misclassification. For instance, if the main reasons for links to be misclassified are that individuals are getting fatigued during the interview, forgetting their connections, or apprehensive about listing their links, then we can expect that false negatives are more likely to be present in $G_n$ than false positives, so that $\rho_1 \geq \rho_0=0$.

The current setup is chosen to simplify the exposition; however, it is flexible enough that it can be used to characterize heterogeneity in the misclassification process. In Appendix \ref{sec_aprob} we consider a generalization of the misclassification probabilities to be covariate-dependent. In this setting, the heterogeneity in the misclassification probabilities is due to differences in the individual's observed types. This represents a desirable alternatively if the researcher believes that the links of certain profiles of individuals are more likely to be misclassified.

An important limitation of the current setting is that it is not suitable to study top-censoring as a source of measurement error in the network. Top-censoring binds the number of links that each individual can report, and thus represents a restriction on the action space and utility that each individual draws from her total number connections. \citet{depaula/richards-shubik/tamer:2018} developed a network formation model of complete information where the total number of links that each individual can establish is bounded. 
Incorporating this source of measurement error into our setting is an important extension that we leave for future research.\footnote{Top-censoring does not represent a drawback of the empirical application considered in this paper, as the reporting cap for the lending behavior was reached in less than 0.02\% of the survey's total respondents, and these observations were not included in our final data. Section \ref{sec5} discusses the construction of the lending network in detail.} 

The following statement is a key observation in our analysis, which relates the observed network statistics $\gamma_{ij,n}$ to the payoff relevant network statistics $\gamma_{ij,n}^\ast$. 

\begin{lemma}
\label{belief_lemma}
If Assumptions \ref{Ass1}-\ref{Ass3} hold, then $\gamma_{ij,n}^\ast=c(\rho_0,\rho_1)+C(\rho_0,\rho_1)\gamma_{ij,n}$ for every $i,j$, where
$$
\gamma_{ij,n}=E\left[\left(G_{ji,n},\frac{1}{n}\sum_{k\ne i}G_{kj,n},\frac{1}{n}\sum_{k\ne i,j}G_{ki,n} G_{kj,n},\frac{1}{n}\sum_{k\ne i,j}(G_{ki,n}+G_{kj,n})\right)'\mid X, \sigma_n\right],
$$
and, for any $r_0, r_1 \geq 0$ such that $r_0 + r_1 < 1$,
\begin{eqnarray*}
c(r_0,r_1)=
-
\left(\begin{array}{cccc}
1&0&0&0\\
0&1&0&0\\
0&0&1&0\\
\end{array}\right)
\left(\begin{array}{cccc}
1-r_0-r_1&0&0&0\\
0&1-r_0-r_1&0&0\\
0&0&(1-r_0-r_1)^2&r_0(1-r_0-r_1)\\
0&0&0&1-r_0-r_1\\
\end{array}\right)^{-1}
\left(\begin{array}{c}
r_0\\
r_0\\
r_0^2\\
r_0
\end{array}\right)
&&
\\
C(r_0,r_1)=
\left(\begin{array}{cccc}
1&0&0&0\\
0&1&0&0\\
0&0&1&0\\
\end{array}\right)
\left(\begin{array}{cccc}
1-r_0-r_1&0&0&0\\
0&1-r_0-r_1&0&0\\
0&0&(1-r_0-r_1)^2&r_0(1-r_0-r_1)\\
0&0&0&1-r_0-r_1\\
\end{array}\right)^{-1}.
&&
\end{eqnarray*}
\end{lemma}

As $\gamma_{ij,n}^\ast =  E \left[ \left(G_{ji,n}^\ast,\frac{1}{n}\sum_{k\ne i}G_{kj,n}^\ast,\frac{1}{n}\sum_{k\ne i,j}G_{ki,n}^\ast G_{kj,n}^\ast\right)' \mid X, \sigma_n  \right]$, we can equivalently write the statement of Lemma \ref{belief_lemma} as 
\begin{eqnarray*}
E\left[ G_{ji,n}^\ast  \mid X, \sigma_n\right] 
&=&
- (1-\rho_0 - \rho_1)^{-1} \rho_0 + (1-\rho_0 - \rho_1)^{-1}E\left[ G_{ji,n}  \mid  X, \sigma_n\right],
\\
\frac{1}{n} \sum_{k\neq j} E\left[ G_{kj,n}^\ast \mid X, \sigma_n\right] 
&=& 
-(1-\rho_0 - \rho_1)^{-1}\rho_0 + (1-\rho_0 - \rho_1)^{-1} \frac{1}{n} \sum_{k \neq j} E\left[ G_{kj,n}  \mid  X, \sigma_n\right],\mbox{ and }
\\
\frac{1}{n} \sum_{k \neq i,j} E\left[ G_{ki,n}^\ast G_{kj,n}^\ast \mid X, \sigma_n \right] 
&=&
\rho_0^2 (1-\rho_0-\rho_1)^{-2} + (1-\rho_0-\rho_1)^{-2} \frac{1}{n} \sum_{k\neq i,j} E\left[G_{ki,n} G_{kj,n}\mid X, \sigma_n\right]  \\&&- \rho_0^2 (1-\rho_0-\rho_1)^{-2} \frac{1}{n} \sum_{k\neq i,j} E\left[G_{ki,n}+G_{kj,n}\mid X, \sigma_n\right].
\end{eqnarray*}

The first three components in $\gamma_{ij,n}$ are the observed analog to the statistics in $\gamma_{ij,n}^\ast$ since they are determined by the observed network $G_{n}$. 
The last component in $\gamma_{ij,n}$ is the sum of the in-degrees of agents $i$ and $j$, and it is the result of controlling for the unobserved network statistics $\frac{1}{n}\sum_{k\ne i, j}G_{ki,n}^\ast G_{kj,n}^\ast$. 
In other words, the last two statistics in $\gamma_{ij,n}$ control for the beliefs about the unobserved network statistic $\frac{1}{n}\sum_{k\ne i, j}G_{ki,n}^\ast G_{kj,n}^\ast$, which is the only nonlinear endogenous factor. The intuition behind this result is similar to the one found in polynomial regression models with mismeasured continuous covariates \citep{hausman/ichimura/newey/powell:1991}.

\begin{remark}[Independent link-specific shocks $\varepsilon_{ij}$]
   \label{remark:independent_types}
   A common assumption invoked in the literature of strategic network formation with network externalities is for $\varepsilon_{ij}$ to be independent across all $i,j$ (see e.g., \citealt{menzel2015strategic,chandrasekhar/lewis:2016, thirkettle:2019,ridder/sheng:2020,sheng:2020}). If we assume $\varepsilon_{ij}$ to be independent across all $i,j$,  Condition (ii) in Assumption \ref{Ass3} is sufficient to control the nonlinear endogenous factor $\frac{1}{n}\sum_{k\ne i,j}G_{ki,n}^\ast G_{kj,n}^\ast$ in terms of its observed counterpart.  This result follows from the conditionally independent formation of links, i.e., $E\left[ G_{ki,n}^\ast G_{kj,n}^\ast \mid X, \sigma_n \right] = E\left[ G_{ki,n}^\ast \mid X, \sigma_n \right]E\left[ G_{kj,n}^\ast \mid X, \sigma_n \right]$. Then, Condition (ii) in Assumption \ref{Ass3} can be used to express $E\left[ G_{ki,n}^\ast  \mid X, \sigma_n\right] = (1-\rho_0 - \rho_1)^{-1} \left(E\left[ G_{ki,n}  \mid  X, \sigma_n\right] -\rho_0\right)$ and similarly for the conditional expectation of $G_{kj,n}^\ast$.
\end{remark}

Assumptions \ref{Ass1}-\ref{Ass3} imply the following relationship between the distributions of $G_{ij,n}$ and $G_{ij,n}^\ast$, which will be used in our identification analysis.
Since we observe $G_{ij,n}$ in the dataset but the outcome of interest is $G_{ij,n}^\ast$, it is crucial to connect these two objects. 
\begin{lemma}\label{Ass3_lemma}
Under Assumptions \ref{Ass1}-\ref{Ass3}, $Pr(G_{ij,n}=1\mid X_{ij},\gamma_{ij,n}, \gamma_{ij,n}^\ast)= \rho_0 Pr(G_{ij,n}^\ast=0\mid X_{ij},\gamma_{ij,n}^\ast)+(1-\rho_1)Pr(G_{ij,n}^\ast=1\mid X_{ij},\gamma_{ij,n}^\ast)$. 
\end{lemma}

\section{Identification Analysis}\label{sec3}
We characterize the identified set based on the joint distribution $P_{0,n}$ of the observed variables $(G_{ij,n},X_{ij},\gamma_{ij,n})$.\footnote{The  identified set can be characterized based on the joint distribution of $\{(G_{ij,n},X_{ij},\gamma_{ij,n}):i,j\in\mathcal{N}_n\}$. However, it is infeasible to estimate the joint distribution of all the dyads in the network from a sample of $n$ agents. Hence, the identified set based on $\{(G_{ij,n},X_{ij},\gamma_{ij,n}):i,j\in\mathcal{N}_n\}$ is not immediately useful for inference. In contrast, $P_{0,n}$ can be estimated from our current sample.} In this section, we treat $\gamma_{ij,n}$ as observed because it can be estimated from the data as follows. For a generic value $x$ in the support of $X_{ij}$, we can define
\begin{eqnarray*}
\hat{p}(x)&=&\frac{1}{n^2}\sum_{i,j}1\{X_{ij}=x\}
\\
\hat\gamma(x)&=&\frac{\frac{1}{n^2}\sum_{i,j}\left(G_{ji,n},\frac{1}{n}\sum_{k}G_{kj,n},\frac{1}{n}\sum_{k}G_{ki,n} G_{kj,n},\frac{1}{n}\sum_{k}(G_{ki,n}+G_{kj,n})\right)'1\{X_{ij}=x\}}{\hat{p}(x)},
\end{eqnarray*}
where $\hat{p}(x)$ is an estimator for $Pr(X_{ij}=x)$ and $\hat\gamma_{ij}=\hat\gamma(X_{ij})$ is an estimator for $\gamma_{ij,n}$.
Then we can estimate the distribution of $(G_{ij,n},X_{ij},\gamma_{ij,n})$ using the empirical distribution of $(G_{ij,n},X_{ij},\hat\gamma_{ij})$.

To formalize our identification analysis, we introduce the following notation.
Denote by $\mathcal{P}^\ast$ the set of joint distributions of $(G_{ij,n},G_{ij,n}^\ast,X_{ij},\gamma_{ij,n},\gamma_{ij,n}^\ast,\varepsilon_{ij})$.
Define the parameter space $\Theta=\mathcal{B}\times\mathcal{R}$, where $\mathcal{B}$ is the parameter space for $\beta_0$ and $\mathcal{R}$ is a subset of $\{(r_0,r_1): r_0,r_1\geq 0,r_0+r_1<1\}$. 
Denote by $\mathcal{P}$ the set of joint distributions of $(G_{ij,n},X_{ij},\gamma_{ij,n})$.

Given Assumptions \ref{Ass1}-\ref{Ass3} and based on the results summarized in Lemmas \ref{lemma1}-\ref{Ass3_lemma}, we impose the following three conditions on the true joint distribution $P_{0,n}^\ast$ of the variables $(G_{ij,n},G_{ij,n}^\ast,X_{ij},\gamma_{ij,n},\gamma_{ij,n}^\ast,\varepsilon_{ij})$ and the true parameter value $\theta_0=(\beta,\rho_0,\rho_1)$.
\begin{condition}\label{independence_assn} 
Under $P^\ast$, the following holds:
(i)  $\varepsilon_{ij}$ is normally distributed with mean zero and variance one. 
(ii) $\varepsilon_{ij}$ and $(X_{ij},\gamma_{ij,n}^\ast)$ are independent.
\end{condition}
\begin{condition}\label{linear_index}
$G_{ij,n}^\ast=1 \left\{(Z_{ij,n}^\ast)'b + \varepsilon_{ij}\geq 0\right\}$ a.s. $P^\ast$, where $Z_{ij,n}^\ast=\left((\gamma_{ij,n}^\ast)',(X_{ij})'\right)'$.
\end{condition}
\begin{condition}\label{misclas_prop}
(i) $P^\ast(G_{ij,n}=1\mid X_{ij},\gamma_{ij,n}, \gamma_{ij,n}^\ast)=r_0P^\ast(G_{ij,n}^\ast=0\mid X_{ij},\gamma_{ij,n}^\ast)+(1-r_1)P^\ast(G_{ij,n}^\ast=1\mid X_{ij},\gamma_{ij,n}^\ast)$. 
(ii) $\gamma_{ij,n}^\ast=c(r_0,r_1)+C(r_0,r_1)\gamma_{ij,n} \; a.s. \; P^\ast$.
\end{condition}

For each element $P$ of $\mathcal{P}$, we are going to define the identified set based on the three conditions above. 
\begin{definition}
   \label{def:identfiedset}
For each distribution $P\in\mathcal{P}$, 
the  identified set $\Theta_I(P)$ is defined as the set of all $\theta=(b,r_0,r_1)$ in $\Theta$ for which there is some joint distribution $P^\ast\in\mathcal{P}^\ast$ such that 
Conditions \ref{independence_assn}, \ref{linear_index}, and \ref{misclas_prop} hold and the distribution of $(G_{ij,n},X_{ij},\gamma_{ij,n})$ induced from $P^\ast$ is equal to $P$.
\end{definition}
The definition of $\Theta_I(P)$ does not depend on $n$, but the identified set $\Theta_I(P_{0,n})$ under the data generating process $P_{0,n}$ can depend on the sample size when the data distribution $P_{0,n}$ depends on $n$. 

Using Definition \ref{def:identfiedset}, we characterize the identified set in Theorem \ref{theorem_identification}. In the next theorem, we use the notation $b=(b_1', b_2')'$ for any generic $b \in \mathcal{B}$, where $b_1$ represents the first three components of $b$ associated with the network externalities and $b_2$ represents the remaining components in $b$ associated with the homophily covariates $X_{ij}$.
\begin{theorem}\label{theorem_identification}
Given a joint distribution $P\in\mathcal{P}$, 
$\Theta_I(P)$ is equal to the set of $\theta\in\Theta$ satisfying  
\begin{equation}\label{moment_TheoremID}
E_P[G_{ij,n}
\mid X_{ij},\gamma_{ij,n}]=
\Psi(\theta,X_{ij},\gamma_{ij,n}),
\end{equation}
where, for a generic value $(x,\check{\gamma}_{ij})$ of $(X_{ij},\gamma_{ij,n})$, we define  
$$
\Psi(\theta,x,\check{\gamma}_{ij})=r_0+(1-r_0-r_1)\Phi((c(r_0,r_1)+C(r_0,r_1)\check{\gamma}_{ij})'b_1+x'b_2).
$$ 
\end{theorem}

Theorem \ref{theorem_identification} states that given a distribution $P$ of $(G_{ij,n},X_{ij},\gamma_{ij,n})$, the identified set is characterized as the set of parameter values that satisfy the moment conditions in Eq. \eqref{moment_TheoremID}.
Notice that under Assumptions \ref{Ass1}-\ref{Ass3}, Conditions \ref{independence_assn}-\ref{misclas_prop} hold for the true joint distribution $P_{0,n}^\ast$ and the true parameter value $\theta_0$. Therefore, we have the following result from Theorem \ref{theorem_identification}.

\begin{corollary}
Under Assumptions \ref{Ass1}-\ref{Ass3}, the true parameter value $\theta_0$ belongs to the set $\theta\in\Theta$ satisfying  \eqref{moment_TheoremID}.
In other words, the set $\theta\in\Theta$ satisfying  \eqref{moment_TheoremID} is an outer identified set (i.e., a superset of the sharp identified set) for $\theta_0$ under Assumptions \ref{Ass1}-\ref{Ass3}.
\end{corollary}

If the links were measured without error, the moment equation in Eq. (\ref{moment_TheoremID}) would degenerate into the model in \cite{leung2015two}: $E_P[G_{ij,n}-\Phi([\gamma_{ij,n}]_{123}'b_1+X_{ij}'b_2)\mid X_{ij},\gamma_{ij,n}]=0$, where $[\gamma_{ij,n}]_{123}$ is a vector composed by the first three components of $\gamma_{ij,n}$. 

Our characterization of the identified set in Theorem \ref{theorem_identification} uses the assumption that $\varepsilon_{ij}$ is normally distributed. However, this condition is invoked for simplicity and is not essential for our analysis. In fact, the results can be easily extended to the case that $\varepsilon_{ij}$ has an absolutely continuous distribution $F_{\varepsilon}(\cdot; \theta_{\varepsilon})$ that is known up to a finite-dimensional parameter $\theta_{\varepsilon} \in \Theta_{\varepsilon}$.
Moreover, in Appendix \ref{sec_a2}, we characterize the identified set as a collection of moment inequalities in a semiparametric framework when the distribution of $\varepsilon_{ij}$ is unknown.

\section{Inference}\label{sec4}

In this section, we construct confidence intervals for $\theta$ based on the identification analysis in Theorem \ref{theorem_identification} and derive its asymptotic coverage when we observe one single network with many agents. As in \cite{leung2015two} and \cite{ridder/sheng:2020}, we use asymptotic arguments based on a symmetric equilibrium.

We derive two confidence intervals for a pre-specified significance level $\alpha\in(0,1)$, and we suggest using $\hat{\mathcal{C}}_n(\alpha)$ introduced in Section \ref{easier_confidence} rather than $CI_n(\alpha)$ introduced in Section \ref{better_confidnece}, because the computation of $\hat{\mathcal{C}}_n(\alpha)$ is much less demanding.
In particular, the computation of $\hat{\mathcal{C}}_n(\alpha)$ only requires us to calculate a quasi-maximum likelihood estimator and its confidence interval for pre-specified grid values of $(r_0,r_1)$. 
On the other hand, the computation of $CI_n(\alpha)$ would require us to evaluate the test statistic that characterizes the confidence interval at every value of $\theta=(b,r_0,r_1)$, and therefore the computational cost of $CI_n(\alpha)$ can be exponential in the number of the (exogenous and endogenous) regressors, determined by the dimension of $\beta$.

\subsection{Confidence Interval through Test Inversion}\label{better_confidnece}
Consider the unconditional sample analog of the moment condition in Eq. \eqref{moment_TheoremID}: 
$$
\hat{m}_n(\theta)=\frac{1}{n^2}\sum_{i,j}\left(G_{ij,n}-
\Psi(\theta,X_{ij},\hat\gamma_{ij})
\right)\zeta_{ij},
$$
where $x_1,\ldots,x_J$ are all the support points for $X_{ij}$ and $\zeta_{ij}=(1\{X_{ij}=x_1\},\ldots,1\{X_{ij}=x_J\})'$.
Note that $\hat{m}_n$ is different from the infeasible sample moment 
$$
{m}_n(\theta)=\frac{1}{n^2}\sum_{i,j}\left(G_{ij,n}-\Psi(\theta,X_{ij},\gamma_{ij,n})\right)\zeta_{ij},
$$ 
because it replaces $\gamma_{ij,n}$ with the first-stage estimator $\hat\gamma_{ij}$.

As we show in the appendix, the influence function for $\hat{m}_n(\theta_0)$ is 
\begin{eqnarray*}
\psi_k(\theta_0)
&=&
\frac{1}{n}\sum_{j\ne k}\left(G_{kj,n}-\rho_0-(1-\rho_0-\rho_1)\Phi(u_{kj}(\theta_0))\right)\zeta_{kj}\\&&-(1-\rho_0-\rho_1)\frac{1}{n^2}\sum_{i,j}\left(\phi(u_{kj}(\theta_0))\beta_1'C(\rho_0,\rho_1) \psi_{\gamma,k, n}(X_{ij})\right)\zeta_{ij},
\end{eqnarray*}
where $u_{kj}(\theta_0) = (c(r_0,r_1)+C(r_0,r_1)\gamma_{kj,n})'b_1+X_{kj}'b_2$, and the component $\psi_{\gamma,k,n}(x)$ is the influence function for the first-stage estimator $\hat{\gamma}_{kj}$, defined by 
\begin{eqnarray*}
\psi_{\gamma,k, n}(x) 
&=& 
\frac{1}{n^2} \sum_{i,j} 
\left(\frac{1\{X_{i,j}=x\}}{\hat{p}(x)} \right) 
\left(  
\begin{array}{c}
0\\
G_{kj,n}-E\left[ G_{kj,n}\mid X, \sigma_n \right]\\
G_{ki,n}G_{kj,n}-E\left[G_{ki,n}G_{kj,n}\mid X, \sigma_n \right]\\
G_{ki,n}+G_{kj,n}- E\left[G_{ki,n}+G_{kj,n}\mid X, \sigma_n \right]\end{array}
\right)
\\&&+
\frac{1}{n} \sum_{i} 
\left(\frac{1\{X_{i,k}=x\}}{\hat{p}(x)} \right) 
\left(  
\begin{array}{c}
G_{ki,n}-E\left[ G_{ki,n}\mid X, \sigma_n \right]\\
0\\
0\\
0
\end{array}
\right).
\end{eqnarray*}
We estimate the variance of $\hat{m}_n(\theta)$ by
$$
\hat{S}(\theta)=\frac{1}{n}\sum_{k=1}^n\hat\psi_k(\theta)\hat\psi_k(\theta)'-\left(\frac{1}{n}\sum_{k=1}^n\hat\psi_k(\theta)\right)\left(\frac{1}{n}\sum_{k=1}^n\hat\psi_k(\theta)\right)',
$$
where $\hat\psi_k(\theta)$ denotes the (uncentered) estimated influence function for $\hat{m}_n(\theta)$, given by 
\begin{eqnarray*}
\hat\psi_k(\theta)
&=&
\frac{1}{n}\sum_{j\ne k}G_{kj,n}\zeta_{kj}
-\frac{1}{n^2}\sum_{l,j}\left(\left.\frac{\partial}{\partial\check{\gamma}_{lj}'}\Psi(\theta,X_{lj},\check{\gamma}_{lj})\right|_{\check{\gamma}_{lj}=\hat\gamma_{lj,n}}\right)\hat\psi_{\gamma,k, n}(X_{lj})
\zeta_{lj},
\end{eqnarray*}
and $\hat\psi_{\gamma,k, n}(x) $ denotes the (uncentered) estimated influence function for $\hat{\gamma}_{kj}$, given by 
$$
\hat\psi_{\gamma,k, n}(x) 
=
\frac{1}{n^2} \sum_{i_{1},j_1} 
\frac{1\{X_{i_{1},j_1}=x\}}{\hat{p}(x)} 
\left(  
\begin{array}{c}
0\\
G_{kj_1}\\
G_{ki_{1}}G_{kj_1}\\
G_{ki_{1}}+G_{kj_1}
\end{array}
\right)
+
\frac{1}{n} \sum_{i_{1}} 
\frac{1\{X_{i_{1},k}=x\}}{\hat{p}(x)} 
\left(  
\begin{array}{c}
G_{ki_{1}}\\
0\\
0\\
0
\end{array}
\right).
$$

We construct a confidence interval for $\theta$ as 
$$
CI_n(\alpha)=\{\theta\in\Theta: n\hat{m}_n(\theta)'\hat{S}(\theta)^{-1}\hat{m}_n(\theta)\leq q(\chi^2_J,1-\alpha)\},
$$
where $q(\chi^2_J,1-\alpha)$ is the $(1-\alpha)$ quantile of a $\chi^2$ distribution with $J$ degrees of freedom. 
The degrees of freedom are determined by the number of points in the support of $X_{ij}$.

The following theorem demonstrates the asymptotic coverage for the confidence interval $CI_n(\alpha)$.

\begin{theorem}\label{theorem_inference}
Suppose that (i) the minimum eigenvalue of $Var(\psi_k(\theta_0)\mid X,\sigma_n)$ is bounded away from zero, and (ii)  $\liminf\min_{x}\hat{p}(x)>0$. 
Under Assumptions \ref{Ass1}-\ref{Ass3}, 
$$
\liminf_{n\rightarrow\infty}Pr(\theta_0\in CI_n(\alpha)\mid X,\sigma_n)\geq 1-\alpha.
$$
\end{theorem}

Condition (i) ensures that $\hat{S}(\theta)$ is non-singular. A similar condition is used in \citet[][Theorem 3]{leung2015two}.
Condition (ii) is imposed to guarantee that $\hat{\gamma}_{ij}$ is a uniformly consistent estimator of $\gamma_{ij,n}$.

The construction of $CI_{n}(\alpha)$ requires us to evaluate the test statistics at every value of $\theta = (b, r_0, r_1)$, which can be computationally intractable.
An alternative approach could focus on computing valid confidence intervals for each component of $\theta$ as $ \min/\max_{\theta \in CI_{n}(\alpha) } \theta_j $ for $j=1, \cdots, d+3$.
However, this method requires the solution of $2(d+3)$ optimization problems given a nonlinear constraint. 
Instead, in Section \ref{easier_confidence}, we conduct valid inference in a computationally feasible manner.

\subsection{Confidence Interval based on Quasi-Maximum Likelihood Estimator}\label{easier_confidence}
In this section, we construct a more computationally feasible (but potentially larger) confidence interval for $\beta$, which relies on the construction of a quasi-maximum likelihood estimator. 

The intuition for this method is the following. Suppose that we knew that $(\rho_0,\rho_1)=(r_0,r_1)$ for a given value $(r_0,r_1)\in\mathcal{R}$, then we could construct a confidence interval $\mathcal{C}_n(\alpha;r_0,r_1)$ for $\beta$ by computing the quasi-maximum likelihood estimator $\hat\beta(r_0,r_1)$ and its estimated asymptotic variance $\widehat{AV}(r_0,r_1)$ in the following way. 
We consider the quasi-maximum likelihood estimator  
$$
\hat\beta(r_0,r_1)=\argmax_{b\in\mathcal{B}}\hat{\boldsymbol{Q}}_n(b,r_0,r_1)
$$
where the feasible objective function is 
\begin{equation}\label{eq:Q_function_defin}
\hat{\boldsymbol{Q}}_n(b,r_0,r_1)=\frac{1}{n^2}\sum_{i,j}\log\left(
\Psi(b,r_0,r_1,X_{ij},\hat\gamma(X_{ij}))^{G_{ij,n}}(1-\Psi(b,r_0,r_1,X_{ij},\hat\gamma(X_{ij})))^{1-G_{ij,n}}\right).
\end{equation}
Define the (uncentered) estimated influence function for the score function by 
\begin{eqnarray*}
\hat{\psi}_{\boldsymbol{Q},k, n}(\theta)
&=&
\frac{1}{n}\sum_{j}(G_{kj,n}-\Psi(\theta,X_{kj},\hat\gamma(X_{kj})))
\boldsymbol{v}_1(\theta,X_{ij},\hat\gamma(X_{kj})) \\ 
&&+
\frac{1}{n^2}\sum_{i,j}
\left\{ 
   \Psi(\theta,X_{ij},\hat\gamma(X_{ij}))\boldsymbol{V}_1(\theta,X_{ij},\hat\gamma(X_{ij}))-\boldsymbol{V}_2(\theta,X_{ij},\hat\gamma(X_{ij}))
\right\}
\hat{\psi}_{\gamma,k, n}(X_{ij}),
\end{eqnarray*}
where
\begin{eqnarray*}
\boldsymbol{v}_1(\theta,X_{ij},\check{\gamma}_{ij})&=&\dfrac{\frac{\partial}{\partial b}\Psi(\theta,X_{ij},\check{\gamma}_{ij})}{\Psi(\theta,X_{ij},\check{\gamma}_{ij})(1-\Psi(\theta,X_{ij},\check{\gamma}_{ij}))},\\
\boldsymbol{v}_2(\theta,X_{ij},\check{\gamma}_{ij})&=&\dfrac{\frac{\partial}{\partial b}\Psi(\theta,X_{ij},\check{\gamma}_{ij})}{1-\Psi(\theta,X_{ij},\check{\gamma}_{ij})},\\
\boldsymbol{V}_1(\theta,X_{ij},\check{\gamma}_{ij})&=&\frac{\partial }{\partial \check{\gamma}_{ij}'}\boldsymbol{v}_1(\theta,X_{ij},\check{\gamma}_{ij}),\mbox{ and }\\
\boldsymbol{V}_2(\theta,X_{ij},\check{\gamma}_{ij})&=&\frac{\partial }{\partial \check{\gamma}_{ij}'}\boldsymbol{v}_2(\theta,X_{ij},\check{\gamma}_{ij}).
\end{eqnarray*}
The asymptotic variance for $\hat\beta(r_0,r_1)$ is estimated by 
$$
\widehat{AV}(r_0,r_1)
=\left(\left.\frac{\partial^2}{\partial b\partial b'}\hat{\boldsymbol{Q}}_n(b,r_0,r_1)\right|_{b=\hat\beta(r_0,r_1)}\right)^{-1}\mathcal{\hat{S}}(r_0,r_1)
\left(\left.\frac{\partial^2}{\partial b\partial b'}\hat{\boldsymbol{Q}}_n(b,r_0,r_1)\right|_{b=\hat\beta(r_0,r_1)}\right)^{-1},
$$
where 
\begin{eqnarray*}
\mathcal{\hat{S}}(r_0,r_1)
&=&
\frac{1}{n}\sum_{k=1}^n\hat{\psi}_{\boldsymbol{Q},k, n}(\hat\beta(r_0,r_1),r_0,r_1)\hat{\psi}_{\boldsymbol{Q},k, n}(\hat\beta(r_0,r_1),r_0,r_1)'
\\&&-
\left(\frac{1}{n}\sum_{k=1}^n\hat{\psi}_{\boldsymbol{Q},k, n}(\hat\beta(r_0,r_1),r_0,r_1)\right)
\left(\frac{1}{n}\sum_{k=1}^n\hat{\psi}_{\boldsymbol{Q},k, n}(\hat\beta(r_0,r_1),r_0,r_1)\right)'.
\end{eqnarray*}
We can construct a $(1-\alpha)$ confidence interval for $\beta$:   
$$
\mathcal{C}_n(\alpha;r_0,r_1)=\left\{b\in\mathcal{B}: n(\hat\beta(r_0,r_1)-b)'\widehat{AV}(r_0,r_1)^{-1}(\hat\beta(r_0,r_1)-b)\leq q(\chi^2_{d+3},1-\alpha)\right\},
$$
where $q(\chi^2_{d+3},1-\alpha)$ is the $(1-\alpha)$ quantile of a $\chi^2$ distribution with ${d+3}$ degrees of freedom.\footnote{We can construct a confidence interval for a subvector $\eta'\beta$ of a given vector $\eta$: 
$$
\eta'\hat\beta(r_0,r_1)\pm q(N(0,1),1-\alpha/2)\sqrt{\frac{\eta'\widehat{AV}(r_0,r_1)\eta}{n}}
$$
where $q(N(0,1),1-\alpha/2)$ is the $(1-\alpha/2)$ quantile of the standard normal distribution. In the same way as $\beta$, we can also take the union over $(r_0,r_1)\in\mathcal{R}$ and construct a feasible confidence interval for $\eta'\beta$.} 

Since the true value of $(\rho_0,\rho_1)$ is unknown, we construct a confidence interval for $\beta$ by taking the union of $\mathcal{C}_n(\alpha;r_0,r_1)$ over $(r_0,r_1)\in\mathcal{R}$: 
\begin{equation*}
   \hat{\mathcal{C}}_n(\alpha)=\bigcup_{(r_0,r_1)\in\mathcal{R}}\mathcal{C}_n(\alpha;r_0,r_1),   
\end{equation*}
where $\mathcal{R}\subseteq \left\{ (r_0, r_1): r_0,r_1 \geq 0, r_0 + r_1 < 1 \right\}$. The next theorem shows that $\hat{\mathcal{C}}_n(\alpha)$ contains the true parameter value with correct asymptotic size. 

Notice that the construction of $\hat{\mathcal{C}}_n(\alpha)$ does not require the value of $(\rho_0, \rho_1)$ to be known, but only to be covered by $\mathcal{R}$. This is a mild condition if we set $\mathcal{R} = \left\{ (r_0, r_1): r_0, r_1 \geq 0 , r_0 +r_1 < 1  \right\}$. Evidently, the larger the set $\mathcal{R}$ is, the less informative the confidence interval $\hat{\mathcal{C}}_n(\alpha)$ will be.
A similar remark has been discussed by \citet[Section 3]{molinari:2008}.\footnote{Notice that it is possible to use the infeasible confidence interval $\mathcal{C}_n(\alpha;r_0,r_1)$ to construct a sensitivity analysis in the spirit of the breakdown frontier, as in \citet{masten/poirer:2020}. In particular,  the criterion in $\mathcal{C}_n(\alpha;r_0,r_1)$ could be used to partition the space $\mathcal{R}$ into the set of values of $(r_0, r_1)$ for which $\beta_0$ is statistically significant and those for which the evidence is inconclusive. We thank an anonymous referee for pointing out this connection.}
Ultimately, the selection of $\mathcal{R}$ should be application-specific as the researcher might have a prior on the upper bounds of $(\rho_0, \rho_1)$. For instance, the researcher might have reasons to believe that in the application at hand, false negatives are more like than false positives, $r_1>r_0$, or the sum of the misclassification probabilities cannot exceed 30\%, i.e., $r_0+r_1 \leq 0.3$.
In Section \ref{sec5}, we assess the informativeness of $\hat{\mathcal{C}}_n(\alpha)$ for different values of the misclassification probabilities.

\begin{theorem}\label{theorem_inference_2}
Suppose that 
(i) $\liminf\min_{x}\hat{p}(x)>0$, 
(ii) $\beta_0$ is in the interior of a compact subset $\mathcal{B}$ of the Euclidean space,  
(iii) $\{((\gamma_{ij,n}^\ast)',X_{ij}')':i,j\}$ is not contained in any proper linear subspace of $\mathbb{R}^{d+3}$, and 
(iv) the minimum eigenvalue of $E\left[\frac{1}{n}\sum_{k=1}^n\psi_{\boldsymbol{Q},k, n}\psi_{\boldsymbol{Q},k, n}'\mid X,\sigma_n\right]$ is bounded away from zero, where 
\begin{eqnarray*}
\psi_{\boldsymbol{Q},k, n}
&=&
\frac{1}{n}\sum_{j}(G_{kj,n}-\Psi(\theta_0,X_{kj},\gamma_{kj,n}))
\boldsymbol{v}_1(\theta_0,X_{ij},\gamma_{ij,n})\\&&+\frac{1}{n^2}\sum_{i,j}\left(E\left[G_{ij,n}\mid X,\sigma_n\right]\boldsymbol{V}_1(\theta_0,X_{ij},\gamma_{ij,n})-\boldsymbol{V}_2(\theta_0,X_{ij},\gamma_{ij,n})\right)\psi_{\gamma,k, n}(X_{ij}).
\end{eqnarray*}
Under Assumptions \ref{Ass1}-\ref{Ass3}, 
$$
\liminf_{n\rightarrow\infty}Pr(\beta_0\in \hat{\mathcal{C}}_n(\alpha)\mid X,\sigma_n)\geq 1-\alpha.
$$
\end{theorem}
Condition (i) is imposed to guarantee that $\hat{\gamma}_{ij}$ is a uniformly consistent estimator of $\gamma_{ij,n}$.
Condition (ii) is a regularity condition and is used to derive the asymptotic distribution of $\hat\beta(\rho_0,\rho_1)$.
Condition (iii) is a rank condition and guarantees that $\beta_0$  is the unique maximizer of the limiting objective function when $(r_0,r_1)=(\rho_0,\rho_1)$.
This assumption guarantees that the equilibrium beliefs about the network statistics $\gamma_{ij,n}^\ast$ have sufficient exogenous variations for any finite $n$. This condition is analogous to the standard rank condition for the identification of discrete choice models (e.g., \citealt{aradillaslopez:2010, leung2015two}). As in \citet[p.187]{leung2015two},  it restricts the equilibrium selection to choose an equilibrium network with non-degenerate $\gamma_{ij,n}^\ast$. 
Condition (iv) ensures that $\widehat{AV}(\rho_0,\rho_1)$ is asymptotically invertible. 

The size property of $\hat{\mathcal{C}}_n(\alpha)$ in Theorem \ref{theorem_inference_2} follows from 
\begin{equation}\label{eq:conv_true_quasiMLE}
\sqrt{n}(\widehat{AV}(\rho_0,\rho_1))^{-1/2}(\hat\beta(\rho_0,\rho_1)-\beta_0)\rightarrow_dN(0,I),
\end{equation}
because $Pr(\beta_0\in \hat{\mathcal{C}}_n(\alpha)\mid X,\sigma_n)\geq Pr(\beta_0\in \mathcal{C}_n(\alpha;\rho_0,\rho_1)\mid X,\sigma_n)$.
Although Eq. (\ref{eq:conv_true_quasiMLE}) is proved in a similar manner to \citet[][Theorem 3]{leung2015two}, it is not a direct implication of \citet[][Theorem 3]{leung2015two} since we do not directly observe the true underlying network $G_{n}^\ast$.

\section{Empirical Illustration}\label{sec5}

In this section, we implement the confidence interval introduced in Section \ref{easier_confidence}  using a social network dataset from rural villages in southern India.
In particular, we investigate the robustness of the empirical results in  \citet{leung2015two} to the presence of misclassified links. 
This dataset was assembled to study the introduction of a micro-finance program (cf. \citealt{banerjee/chandrasekhar/duflo/jackson:2013} and  \citealt{jackson/rodriguezbarraquer/tan:2012}) and contains information on household and individual characteristics and data on social networks.

Among the different dimensions of social relationships contained in the dataset, we follow \citet{leung2015two} and focus on a lending network. 
This network measures individuals' willingness to lend money. We construct the directed links in our analysis using the raw answers to the following survey question: ``\emph{Who do you trust enough that if he/she needed to borrow Rs. 50 for a day you would lend it to him/her?}''\footnote{Our analysis does not use the adjacency matrices constructed by \citet{banerjee/chandrasekhar/duflo/jackson:2013}, which are readily available with the original dataset since they are \emph{undirected.} Instead, we use the social connections reported by the individuals to construct the \emph{directed} adjacency matrices. \cite{jackson/rodriguezbarraquer/tan:2012} discusses this point in detail.}

\cite{jackson/rodriguezbarraquer/tan:2012} discuss potential measurement error issues that might be present in this dataset.
In particular, they argue that the most likely type of measurement error to appear in this dataset is the misclassification of true links as non-existent, i.e., false negatives. This is because under the structure of the survey questions, which ask individuals about actual actions (such as lending or borrowing money) rather than perceived relationships, individuals are more likely to forget existing interactions than to imagine non-existent ones. Among the potential sources discussed by \cite{jackson/rodriguezbarraquer/tan:2012} are (i) individuals not remembering their connections, (ii) people getting wearied during interviews, and (iii) top-censoring the number of social connections that individuals could report.\footnote{We assume that the set of individuals observed in our empirical application represents the complete list of nodes in the network. Network subsampling is an interesting and challenging setting that is beyond the scope of this paper. \citet{liu:2013, chandrasekhar/lewis:2016,ata/belloni/candogan:2018},  and \citet{thirkettle:2019} have analyzed the effects of using sampled network data when estimating social interaction models or network externalities.}  In terms of the lending network, top-censoring does not represent a concern as the reporting cap was not reached by any of the individuals included in this sample. We explain this in detail below.  

First, in this empirical illustration, we focus on the false-negative-only case (i.e., $r_0=0$). Then we also consider the false-positive-only case (i.e., $r_1=0$) and a fully flexible specification that includes both false positives and negatives. 

We examine the relative importance that homophily on observed attributes and endogenous beliefs about trustworthiness have in the formation of links in a lending network. Regarding the preferences for homophily, we study homophily relations on gender, caste, language, religion, and family relationships. In terms of religion, Hinduism represents the large majority. 
Following \citet{leung2015two}, we avoid multicollinearity by using only the villages where the non-Hindu minorities have at least a 10\% representation; there are nine of these villages in total.

We aggregate the information across these nine villages to implement the quasi-maximum likelihood estimator described in Section \ref{easier_confidence}.\footnote{In this empirical application, the quasi-maximum likelihood estimator is defined as $\hat{\beta}(r_0,r_1) = \argmax_{b \in \mathcal{B}} \sum_{m=1}^{M} \hat{Q}_{m}(b, r_0, r_1)$ where $m$ denotes the $m$-th village and $\hat{Q}_{m}(b, r_0, r_1)$ denotes the objective function for $m$-th village as defined in \eqref{eq:Q_function_defin}.} The total number of individuals in our sample is $2,031$. Our asymptotic framework is based on the number of agents growing large and not the number of villages. Notice that an asymptotic analyses with many small networks may not be suitable for this empirical example, as the number of villages is much smaller than the number of individuals.

The average in-degree and out-degree across the individuals is equal to 0.951. The largest in-degree value observed for an individual is 10, which means that one individual in the sample was listed 10 times as someone to whom other individuals would be willing to lend money. Meanwhile, the maximum out-degree observed across the individuals is equal to 4. In other words, the maximum number of people to whom a single individual within our sample was willing to lend money is equal to 4. This value suggests that the recording cap for the lending network, given by 5, was not attained in this sample. Thus, top-censoring does not represent a source of concern for the misclassification of links in this application.

The direct link  $G_{ij,n}$, for any distinct individuals $i$ and $j$, is recorded to be equal to $1$ if individual $i$ lists $j$ as someone to whom he or she is willing to lend Rs. 50, and $0$ otherwise. We allow for $G_{ij,n}$ to be potentially misclassified. 
In the vector of observed attributes $X_{ij}$, we include individual $i$- and individual $j$-specific regressors, such as age, caste, gender, religion, and an indicator for whether or not $i$ and $j$ are heads of their household, as well as the controls for homophily described above.

In the vector of endogenous network statistics $\gamma_{ij,n}^{*}$, we consider the conditional expectation of the following factors:
(i) $G_{ji,n}^{*}$, which accounts for the value of reciprocation; (ii) $n^{-1}\sum_{k \neq i} G_{kj,n}^{*}$, which measures the share of people willing to lend money to $j$; and (iii) $n^{-1}\sum_{k \neq i, j} G_{ki,n}^{*} G_{kj,n}^{*}$, which is the supported trust or share of individuals that are willing to lend to both $i$ and $j$. We account for the misclassification on the endogenous network statistics via Lemma \ref{belief_lemma}. As a first stage estimator, we use the frequency estimator described in Section \ref{sec3}.

In order to examine the effects of potentially misclassified links on the estimation of the structural parameters of a network formation model, we allow for the possibility that true underlying links in $G_n^\ast$ are misclassified as missing in $G_n$. In particular, we consider the following scenarios for the misclassification probabilities: we set the space $\mathcal{R} = \left\{ (r_0, r_1): r_0=0, 0 \leq r_1 \leq \bar{\mathcal{R}}_1 \right\}$ with $\bar{\mathcal{R}}_1 \in \{0,0.05, 0.1,0.2,0.3,0.4,0.5\}$.\footnote{Additional results with misclassification probabilities $r_1$ of up to $90\%$ are reported in Appendix \ref{S:additionaltables}.} Notice that the set $\mathcal{R}$ does not restrict the true misclassification probabilities $(\rho_0, \rho_1)$ to be known. Instead, it represents an upper bound on the amount of misclassification that might affect the network links, i.e.,  the probability of a false negative link is up to $\bar{\mathcal{R}}_1$. This prior allows us to construct confidence intervals $\hat{C}_n(\alpha)$ that are robust up to an $\bar{\mathcal{R}}_1$ probability of misclassifying links as negative. 

Table \ref{table:union} presents 95\% confidence intervals using the quasi-maximum likelihood estimator of the network statistics, the homophily parameters, and the constant term. In other words, we report the confidence intervals $\hat{\mathcal{C}}_n(\alpha)=\bigcup_{0 \leq r_1 \leq \bar{\mathcal{R}}_1}\mathcal{C}_n(\alpha;0,r_1)$ for each of the parameters of interest. Column (1) of Table \ref{table:union} presents the confidence intervals for the no-misclassification case, i.e., $r_0=r_1=0$. This scenario corresponds to the empirical analysis in \citet{leung2015two}, and we use it as our baseline case.\footnote{Unlike the empirical specification in \citet{leung2015two}, we do not include in-degree or out-degree statistics weighted by caste or religion as part of the factors that capture network externalities. We do this for simplicity, but the results can be easily generalized to \citet{leung2015two} complete specification.} These results indicate that the network externalities given by reciprocation, in-degree, and supported trust are important factors that determine the willingness of an individual to lend money.  In other words, an individual within the network is more willing to lend money to someone else if that trust is reciprocated, if a large fraction of individuals in the network trust that individual, and if they have common connections that are willing to lend money to both individuals. There is also evidence that individuals present preferences for homophily on gender, religion, caste, and for being relatives when making money lending decisions. The intercept term takes a large negative value to rationalize the fact that the observed network is sparse, and hence, many of the connections are not established.

Columns (2)-(7)  of Table \ref{table:union} present 95\% confidence intervals for the parameter estimates that account for up to $r_1 \leq \bar{\mathcal{R}}_1$ misclassification probabilities with $\bar{\mathcal{R}}_1 \in \{0.05, 0.1,0.2,0.3,0.4,0.5\}$. Although both the network externalities and homophily factors remain significant components driving the formation of a link on trust networks, the confidence intervals become wider as the misclassification of links becomes more likely in the data. Nonetheless, this increase in the width of the confidence intervals remains relatively small. Table \ref{table:ratiolength} provides further evidence about this insight. 

In Table \ref{table:ratiolength}, we compare the length of 95\% confidence intervals that take misclassification into account, with the length of the confidence intervals computed under the assumption of no misclassification, i.e.,  $|\hat{\mathcal{C}}_{n}(\alpha)|/|\mathcal{C}_{n}(\alpha, 0, 0)|$ when $\bar{\mathcal{R}}_1 \in \{0.05, 0.1,0.2,0.3,0.4,0.5\}$. Given a misclassification probability of up to 5\%, the coefficient for in-degree presents the largest increase in interval length, which is $1.17$ times  (or equivalently 17\%) larger than the baseline confidence interval. With a misclassification probability of up to 20\%, the largest increase in the length of the confidence intervals is $1.67$ times (or equivalently 67\%) larger than the confidence intervals that do not take measurement error into account. Finally, in the case that the misclassification probability is at most 50\%, the largest increase in the length of confidence intervals is $2.95$ times larger than the benchmark confidence intervals, and the second-largest increase is $2.70$ times larger. 

As in \citet{leung2015two}, this empirical illustration suggests that the network externalities and homophily factors affect the lending decisions of the individuals in the network. These results are robust to the misclassification of existent links in $G_n^\ast$ as negatives. Furthermore, we use different scenarios for the misclassification probabilities to compare the length of 95\% confidence intervals computed under our method to the length of confidence intervals that assume no link misclassification. The analysis suggests that with a misclassification probability of up to 50\%, the 95\% confidence intervals are at most $2.95$ times larger than the length of the baseline confidence intervals. 

\subsection{Alternative Designs: False Positives and Fully Flexible}
In this subsection, we consider the two alternative designs. In the first design, we consider the case in which only non-existing links in $G_n^\ast$ can be misclassified as existing in $G_n$, i.e.,  false positives. In the second design, we consider a fully flexible specification in which both false positives and negatives can be present.

To analyze the false-positive-only case, we set the space $\mathcal{R} = \left\{ (r_0, r_1): 0 \leq r_0 \leq \bar{\mathcal{R}}_0,  r_1 =0 \right\}$ with the upper bound given by  $\bar{\mathcal{R}}_0 \in \{0,0.01, 0.02,0.03\}$. The probability of 3\% represents a conservative upper bound on $\rho_0$. In other words, for this empirical application, the model rules out that the probability of false positives can be larger than 3\%.\footnote{In particular, the semiparametric analysis in Appendix C yields moment inequalities that bound the misclassification probabilities $(\rho_0, \rho_1)$. We use an upper bound $\max_{k \in \mathcal{N}_n} \frac{1}{n} \sum_{i=1}^{n} G_{ik}\approx 0.03$ for $\rho_0$ and  $\max_{k \in \mathcal{N}_n} \frac{1}{n} \sum_{i=1}^{n} \left(1-G_{ik}\right)\approx 1$ for $\rho_1$.} Parallel to the false-negative-only case, the set  $\mathcal{R}$ does not restrict the true misclassification probabilities $(\rho_0, \rho_1)$ to be known. Instead,  $\bar{\mathcal{R}}_0$ represents an upper bound on the false positive rate. Table \ref{table:unionR0} presents 95\% confidence intervals for the estimates of the parameters of interest, which are computed as $\hat{\mathcal{C}}_n(\alpha)=\bigcup_{0 \leq r_0 \leq \bar{\mathcal{R}}_0}\mathcal{C}_n(\alpha;r_0,0)$. In Column (1) of Table \ref{table:unionR0}, we report the confidence intervals for the no-misclassification case.  Columns (2)-(6) of Table \ref{table:unionR0} present 95\% confidence intervals for the parameter estimates that account for up to $r_0 \leq \bar{\mathcal{R}}_0$. Given a misclassification probability of up to 1\%, the confidence intervals become wider, and the in-degree statistic is no longer statistically significant to explain the formation of the network links. If the probability of misclassifying links as positive is up to 2\%, only the factors accounting for reciprocation, same sex, and same family remain statistically significant at a 95\% significance level. The same conclusion is drawn when the misclassification probability takes the value of 3\%.

Table \ref{table:ratiolengthR0} compares the length of 95\% confidence intervals that take misclassification of non-existent links into account with the length of the confidence intervals computed under the assumption of no link misclassification. In particular, we compute $|\hat{\mathcal{C}}_{n}(\alpha)|/|\mathcal{C}_{n}(\alpha, 0, 0)|$ when $\bar{\mathcal{R}}_0 \in \{0.01, 0.02,0.03\}$. The evidence suggests that, with an upper bound of 1\% probability, the confidence intervals can be $6.06$ times wider in length than the baseline confidence intervals. With a misclassification probability of up to 2\%, the largest increase in the length of the confidence intervals is $7.30$ times larger than the confidence intervals that do not take measurement error into account. Finally, in the extreme case that the misclassification probability is up to 3\%, the 95\% confidence intervals are at most $14.25$ times larger than the length of the baseline confidence intervals.

In the final design, we consider a setting that accounts for both false positives and negatives. Tables \ref{table:unionmixed005} and \ref{table:unionmixed01} present the confidence intervals for the cases in which the false positive rate $r_0$ is bounded up to $1\%$ and $2\%$, respectively, and $r_1 \leq 40\%$. The results indicate that when the misclassifications probabilities are $r_0 \leq 1\%$ and $r_1 \leq 5\%$, all the factors other than reciprocity and same language are statistically significant factors at a 95\% significance level. This conclusion does not change as the upper bound for the probability of misclassifying links as negatives increases up to 40\%; however, most of the confidence intervals become wider to account for the higher uncertainty. In contrast, when the upper bound of the probability of misclassifying links as positives increases to 2\%, only the factors associated with reciprocity, same sex, same family, and the constant term remain statistically significant at a 95\% for any $r_1\leq 4\%$. 

Table \ref{table:ratiolengthmixed005} and \ref{table:ratiolengthmixed01} compare the length of 95\% confidence intervals that take both types of misclassification into account with the length of the confidence intervals computed under the assumption of no link misclassification. The results show that when $r_0\leq 1\%$ and $r_1\leq 4\%$ the confidence intervals can be even $8.21$ times wider in length than the baseline confidence intervals. In contrast when $r_0\leq 2\%$ and $r_1\leq 4\%$,  the confidence intervals can be even $12.57$ times wider in length than the baseline confidence intervals. Relative to both the false-negative-only and false-positive-only cases, the width of the confidence intervals are wider as it accounts for the two-sided link misclassification.

These results suggest that the estimates of the parameters in the network formation model might be sensitive to misclassifying non-existent links as real links, even at small probabilities. One explanation behind these outcomes could be the sparsity of the networks considered in our empirical application. In other words, due to the reduced number of links formed in the network, the effects of misclassifying links are asymmetric and more pervasive for the false-positive case.

\section{Conclusion}\label{sec6}
We study a network formation model with incomplete information and potentially misclassified links.  We propose a novel approach for analyzing network formation models, which is robust to link misclassification. 
In the presence of link misclassification, we characterize the identified set for the preference parameters associated with homophily and network externalities. Based on the identification result, we develop an inference method that is valid when a single large network is available. In an empirical application, we compute conservative confidence intervals that are robust to link misclassification using a lending network from rural villages in Karnataka, India. 
Our results suggest that reciprocity, same gender, and same family are statistically significant factors that explain the formation of a lending network even under link misclassification.

\clearpage
\begin{table}[ht] 
\centering
\caption{95\% confidence intervals $\hat{\mathcal{C}}_n(\alpha)$ with  $r_0=0$.}
   \label{table:union}
   \begin{tabular}{
    l
    >{{[}} 
    S[table-format = -2.3,table-space-text-pre={[}]
    @{,\,} 
    S[table-format = -2.3,table-space-text-post={]}]
    <{{]}} 
    >{{[}} 
    S[table-format = -2.3,table-space-text-pre={[}]
    @{,\,} 
    S[table-format = -2.3,table-space-text-post={]}]
    <{{]}} 
    >{{[}} 
    S[table-format = -2.3,table-space-text-pre={[}]
    @{,\,} 
    S[table-format = -2.3,table-space-text-post={]}]
    <{{]}} 
    >{{[}} 
    S[table-format = -2.3,table-space-text-pre={[}]
    @{,\,} 
    S[table-format = -2.3,table-space-text-post={]}]
    <{{]}} 
}
\hline
&\multicolumn{2}{c}{$r_1=0$}&\multicolumn{2}{c}{$r_1\leq0.05$}&\multicolumn{2}{c}{$r_1\leq 0.1$} & \multicolumn{2}{c}{$r_1\leq 0.2$}\cr
&  \multicolumn{2}{c}{(1)}& \multicolumn{2}{c}{(2)}&\multicolumn{2}{c}{(3)}& \multicolumn{2}{c}{(4)}  \cr \hline\hline
Reciprocation   & 1.342  &   1.676    & 1.314   &   1.676      & 1.283   &   1.676       & 1.213   &   1.676     \cr
In degree       & 26.361 &   33.099   & 25.226  &   33.099     & 24.104  &   33.099      & 21.852  &   33.099    \cr
Supported trust & 59.482 &   110.258  & 56.958  &   110.258    & 53.072  &   110.258     & 44.181  &   110.258   \cr
Constant        & -3.896 &   -3.536   & -3.896  &   -3.520     & -3.896  &   -3.506      & -3.896  &   -3.474    \cr
Same religion   & 0.348  &   0.492    & 0.348   &   0.495      & 0.348   &   0.497       & 0.348   &   0.504     \cr
Same sex        & 0.565  &   0.705    & 0.565   &   0.711      & 0.565   &   0.717       & 0.565   &   0.730     \cr
Same caste      & 0.195  &   0.309    & 0.195   &   0.312      & 0.195   &   0.314       & 0.195   &   0.319     \cr
Same language   & -0.008 &   0.077    & -0.009  &   0.077      & -0.010  &   0.077       & -0.011  &   0.078     \cr
Same family     & 1.308  &   1.537    & 1.308   &   1.572      & 1.308   &   1.608       & 1.308   &   1.685     \cr \hline
&  \multicolumn{2}{c}{}  & \multicolumn{2}{c}{$r_1\leq 0.3$}&  \multicolumn{2}{c}{$r_1\leq 0.4$}& \multicolumn{2}{c}{$r_1\leq 0.5$} \cr
& \multicolumn{2}{c}{}  & \multicolumn{2}{c}{(5)}&  \multicolumn{2}{c}{(6)} & \multicolumn{2}{c}{(7)} \cr \hline\hline
Reciprocation   & \multicolumn{2}{c}{}  &  1.134  & 1.676      &  1.042  & 1.676      &  0.938  & 1.676    \cr
In degree       & \multicolumn{2}{c}{}  &  19.584 & 33.099     &  17.270 & 33.099     &  14.907 & 33.099   \cr
Supported trust & \multicolumn{2}{c}{}  &  35.230 & 110.258    &  26.994 & 110.258    &  19.644 & 110.258  \cr
Constant        & \multicolumn{2}{c}{}  &  -3.896 & -3.439     &  -3.896 & -3.399     &  -3.896 & -3.351   \cr
Same religion   & \multicolumn{2}{c}{}  &  0.348  & 0.513      &  0.348  & 0.525      &  0.348  & 0.540    \cr
Same sex        & \multicolumn{2}{c}{}  &  0.565  & 0.744      &  0.565  & 0.761      &  0.565  & 0.781    \cr
Same caste      & \multicolumn{2}{c}{}  &  0.195  & 0.326      &  0.195  & 0.333      &  0.195  & 0.343    \cr
Same language   & \multicolumn{2}{c}{}  &  -0.013 & 0.078      &  -0.015 & 0.079      &  -0.018 & 0.079    \cr
Same family     & \multicolumn{2}{c}{}  &  1.308  & 1.771      &  1.308  & 1.867      &  1.308  & 1.983    \cr \hline
\end{tabular}

   \caption*{\footnotesize{Note: $\hat{\mathcal{C}}_n(\alpha)$ is computed in Columns (2)-(7) as $\cup_{r_1\leq \bar{\mathcal{R}}_1} \mathcal{C}_n(\alpha; 0, r_1)$, with $\bar{\mathcal{R}}_1\in \{0.05, 0.1, 0.2, 0.3, 0.4, 0.5\}$.}}
\end{table}

\begin{table}[ht]       
\centering
\caption{Ratio of lengths of 95\% confidence intervals, $|\hat{\mathcal{C}}_n(\alpha)|/|\mathcal{C}_n(\alpha, 0,0)|$.}
   \label{table:ratiolength}
   \begin{tabular}{ l c c c c c c}
\hline
                & $r_1\leq 0.05$ &  $r_1\leq 0.1$  &  $r_1\leq 0.2$ &  $r_1\leq 0.3$   &  $r_1\leq 0.4$    & $r_1\leq 0.5$ \\ 
                & (1)            &    (2)          &    (3)         &  (4)             & (5)               &  (6)  \\ \hline\hline 
Reciprocation   &  1.085&1.178&1.385&1.624&1.899&2.210   \\
In degree       &  1.169&1.335&1.669&2.006&2.349&2.700   \\
Supported trust &  1.050&1.126&1.301&1.478&1.640&1.785   \\
Constant        &  1.043&1.084&1.170&1.267&1.380&1.513   \\
Same religion   &  1.017&1.037&1.085&1.146&1.225&1.331   \\
Same sex        &  1.040&1.083&1.175&1.280&1.398&1.538   \\
Same caste      &  1.021&1.042&1.088&1.143&1.210&1.295   \\
Same language   &  1.008&1.017&1.038&1.063&1.095&1.135   \\
Same family     &  1.153&1.311&1.648&2.020&2.439&2.946   \\\hline
\end{tabular}

\end{table}

\begin{table}[ht] 
\centering
\caption{95\% confidence intervals $\hat{\mathcal{C}}_n(\alpha)$ with  $r_1=0$.}
   \label{table:unionR0}
   \begin{tabular}{
    l
    >{{[}} 
    S[table-format = -2.3,table-space-text-pre={[}]
    @{,\,} 
    S[table-format = -2.3,table-space-text-post={]}]
    <{{]}} 
    >{{[}} 
    S[table-format = -2.3,table-space-text-pre={[}]
    @{,\,} 
    S[table-format = -2.3,table-space-text-post={]}]
    <{{]}} 
    >{{[}} 
    S[table-format = -2.3,table-space-text-pre={[}]
    @{,\,} 
    S[table-format = -2.3,table-space-text-post={]}]
    <{{]}} 
    >{{[}} 
    S[table-format = -3.3,table-space-text-pre={[}]
    @{,\,} 
    S[table-format = -2.3,table-space-text-post={]}]
    <{{]}} 
}
\hline
&\multicolumn{2}{c}{$r_0=0$}&\multicolumn{2}{c}{$r_0\leq0.01$}&\multicolumn{2}{c}{$r_0\leq 0.02$} & \multicolumn{2}{c}{$r_0\leq 0.03$}\cr
&  \multicolumn{2}{c}{(1)}& \multicolumn{2}{c}{(2)}&\multicolumn{2}{c}{(3)}& \multicolumn{2}{c}{(4)}  \cr \hline\hline
Reciprocation   &1.342  &1.676  &1.342  &2.033  &1.342  &2.303  &1.280   &2.503     \cr 
In degree       &26.361 &33.099 &-4.617 &36.248 &-11.529&37.635 &-27.082 &50.532    \cr 
Supported trust &59.482 &110.258&35.522 &153.452&-77.588&233.503&-302.054&421.456   \cr 
Constant        &-3.896 &-3.536 &-3.896 &-2.713 &-4.849 &-2.713 &-5.526  &-2.713    \cr 
Same religion   &0.348  &0.492  &0.145  &0.492  &-0.040 &0.492  &-0.156  &0.492     \cr 
Same sex        &0.565  &0.705  &0.323  &0.705  &0.323  &0.787  &0.321   &0.895     \cr 
Same caste      &0.195  &0.309  &0.061  &0.438  &-0.194 &0.438  &-0.349  &0.438     \cr 
Same language   &-0.008 &0.077  &-0.115 &0.077  &-0.115 &0.232  &-0.115  &0.304     \cr 
Same family     &1.308  &1.537  &1.308  &2.047  &1.308  &2.631  &1.308   &2.954     \cr \hline
\end{tabular}

   \caption*{\footnotesize{Note: $\hat{\mathcal{C}}_n(\alpha)$ is computed in Columns (2)-(4) as $\cup_{r_0\leq \bar{\mathcal{R}}_0} \mathcal{C}_n(\alpha; r_0, 0)$, with $\bar{\mathcal{R}}_0\in \{0.01, 0.02, 0.03\}$.}}
\end{table}

\begin{table}[ht]       
   \centering
   \caption{Ratio of lengths of 95\% confidence intervals, $|\hat{\mathcal{C}}_n(\alpha)|/|\mathcal{C}_n(\alpha, 0,0)|$.}
      \label{table:ratiolengthR0}
      \begin{tabular}{ l c c c }
    \hline
                    & $r_0\leq 0.01$ &  $r_0\leq 0.02$  &  $r_0\leq 0.03$  \\ 
                    & (1)            &    (2)          &    (3)            \\ \hline\hline 
Reciprocation       &2.071&2.881&3.664\\
In degree           &6.065&7.297&11.519\\
Supported trust     &2.323&6.127&14.249\\
Constant            &3.281&5.922&7.801\\
Same religion       &2.407&3.693&4.493\\
Same sex            &2.732&3.317&4.102\\
Same caste          &3.306&5.544&6.904\\
Same language       &2.244&4.062&4.900\\
Same family         &3.225&5.772&7.185\\\hline
\end{tabular}

\end{table}

\begin{sidewaystable}[ht] 
\centering
\caption{95\% confidence intervals $\hat{\mathcal{C}}_n(\alpha)$.}
   \begin{tabular}{
    l
    >{{[}} 
    S[table-format = -2.3,table-space-text-pre={[}]
    @{,\,} 
    S[table-format = -2.3,table-space-text-post={]}]
    <{{]}} 
    >{{[}} 
    S[table-format = -2.3,table-space-text-pre={[}]
    @{,\,} 
    S[table-format = -2.3,table-space-text-post={]}]
    <{{]}} 
    >{{[}} 
    S[table-format = -2.3,table-space-text-pre={[}]
    @{,\,} 
    S[table-format = -2.3,table-space-text-post={]}]
    <{{]}} 
    >{{[}} 
    S[table-format = -2.3,table-space-text-pre={[}]
    @{,\,} 
    S[table-format = -2.3,table-space-text-post={]}]
    <{{]}} 
    >{{[}} 
    S[table-format = -2.3,table-space-text-pre={[}]
    @{,\,} 
    S[table-format = -2.3,table-space-text-post={]}]
    <{{]}} 
}
\hline
& \multicolumn{2}{c}{\shortstack{$r_0\leq 0.01$ \\ $r_1\leq 0.05$}} & \multicolumn{2}{c}{\shortstack{$r_0\leq 0.01$ \\ $r_1\leq 0.1$}} & \multicolumn{2}{c}{\shortstack{$r_0\leq 0.01$ \\ $r_1\leq 0.2$}} & \multicolumn{2}{c}{\shortstack{$r_0\leq 0.01$ \\ $r_1\leq 0.3$}} & \multicolumn{2}{c}{\shortstack{$r_0\leq 0.01$ \\ $r_1\leq 0.4$}} \cr   
&  \multicolumn{2}{c}{(1)}& \multicolumn{2}{c}{(2)}&\multicolumn{2}{c}{(3)}& \multicolumn{2}{c}{(4)} & \multicolumn{2}{c}{(5)}  \cr \hline\hline                
Reciprocation   &1.342  &2.033      &1.342   &2.109     &1.342   &2.109     &1.342  &2.109      &1.342  &2.109    \cr 
In degree       &-7.011 &37.492     &-11.421 &40.397    &-11.421 &40.397    &-11.421&40.397     &-11.421&40.397   \cr 
Supported trust &47.250 &145.528    &46.701  &145.528   &46.701  &145.528   &45.733 &145.528    &35.044 &145.528  \cr 
Constant        &-3.896 &-2.594     &-3.896  &-2.426    &-5.220  &-2.426    &-5.274 &-2.426     &-5.387 &-2.426   \cr 
Same religion   &0.116  &0.492      &0.103   &0.574     &0.103   &0.848     &0.103  &0.855      &0.103  &0.889    \cr 
Same sex        &0.326  &0.708      &0.326   &0.829     &0.326   &1.281     &0.326  &1.360      &0.326  &1.455    \cr 
Same caste      &0.047  &0.474      &0.047   &0.556     &0.047   &0.640     &0.047  &0.667      &0.047  &0.700    \cr 
Same language   &-0.121 &0.077      &-0.121  &0.109     &-0.121  &0.280     &-0.121 &0.293      &-0.121 &0.316    \cr 
Same family     &1.308  &2.156      &1.308   &2.388     &1.308   &2.604     &1.308  &2.804      &1.308  &3.016    \cr \hline
\end{tabular}

   \label{table:unionmixed005}
   \caption*{\footnotesize{Note: $\hat{\mathcal{C}}_n(\alpha)$ is computed in Columns (1)-(5) as $\cup_{r_0\leq \bar{\mathcal{R}}_0, r_1\leq \bar{\mathcal{R}}_1} \mathcal{C}_n(\alpha; r_0, r_1)$, with $\bar{\mathcal{R}}_0=0.01$ and $\bar{\mathcal{R}}_1\in \{0.05, 0.1, 0.2, 0.3, 0.4\}$.}}
\end{sidewaystable}

\begin{table}[ht]       
   \centering
   \caption{Ratio of lengths of 95\% confidence intervals, $|\hat{\mathcal{C}}_n(\alpha)|/|\mathcal{C}_n(\alpha, 0,0)|$.}
      \begin{tabular}{l c c c c c}
\hline
& \shortstack{$r_0\leq 0.01$ \\ $r_1\leq 0.05$} & \shortstack{$r_0\leq 0.01$ \\ $r_1\leq 0.1$} & \shortstack{$r_0\leq 0.01$ \\ $r_1\leq 0.2$} & \shortstack{$r_0\leq 0.01$ \\ $r_1\leq 0.3$} & \shortstack{$r_0\leq 0.01$ \\ $r_1\leq 0.4$} \\
&  (1)                         & (2)                        & (3)                         & (4)                          & (5)                         \\ \hline
Reciprocation    &2.070&2.298&2.298&2.298&2.298 \\ 
In degree        &6.605&7.690&7.690&7.690&7.690 \\ 
Supported trust  &1.936&1.946&1.946&1.965&2.176 \\ 
Constant         &3.610&4.077&7.749&7.898&8.211 \\ 
Same religion    &2.610&3.272&5.171&5.218&5.457 \\ 
Same sex         &2.728&3.591&6.822&7.389&8.062 \\ 
Same caste       &3.755&4.469&5.211&5.447&5.740 \\ 
Same language    &2.323&2.696&4.700&4.844&5.124 \\ 
Same family      &3.700&4.713&5.657&6.529&7.452 \\ \hline
\end{tabular}

      \label{table:ratiolengthmixed005}
\end{table}

\begin{sidewaystable}[ht] 
\centering
\caption{95\% confidence intervals $\hat{\mathcal{C}}_n(\alpha)$.}
\begin{tabular}{
    l
    >{{[}} 
    S[table-format = -2.3,table-space-text-pre={[}]
    @{,\,} 
    S[table-format = -2.3,table-space-text-post={]}]
    <{{]}} 
    >{{[}} 
    S[table-format = -2.3,table-space-text-pre={[}]
    @{,\,} 
    S[table-format = -2.3,table-space-text-post={]}]
    <{{]}} 
    >{{[}} 
    S[table-format = -2.3,table-space-text-pre={[}]
    @{,\,} 
    S[table-format = -2.3,table-space-text-post={]}]
    <{{]}} 
    >{{[}} 
    S[table-format = -2.3,table-space-text-pre={[}]
    @{,\,} 
    S[table-format = -2.3,table-space-text-post={]}]
    <{{]}} 
    >{{[}} 
    S[table-format = -2.3,table-space-text-pre={[}]
    @{,\,} 
    S[table-format = -2.3,table-space-text-post={]}]
    <{{]}} 
}
\hline
& \multicolumn{2}{c}{\shortstack{$r_0\leq 0.02$ \\ $r_1\leq 0.05$}} & \multicolumn{2}{c}{\shortstack{$r_0\leq 0.02$ \\ $r_1\leq 0.1$}} & \multicolumn{2}{c}{\shortstack{$r_0\leq 0.02$ \\ $r_1\leq 0.2$}} & \multicolumn{2}{c}{\shortstack{$r_0\leq 0.02$ \\ $r_1\leq 0.3$}} & \multicolumn{2}{c}{\shortstack{$r_0\leq 0.02$ \\ $r_1\leq 0.4$}} \cr   
&  \multicolumn{2}{c}{(1)}& \multicolumn{2}{c}{(2)}&\multicolumn{2}{c}{(3)}& \multicolumn{2}{c}{(4)} & \multicolumn{2}{c}{(5)}  \cr \hline\hline                
Reciprocation   &1.342  &2.382      &1.342  &2.492  &1.342  &2.492  &1.342  &2.492  &1.342  &2.492       \cr 
In degree       &-14.070&40.397     &-14.679&40.397 &-14.679&40.397 &-14.679&40.397 &-14.679&40.397      \cr 
Supported trust &-92.558&250.961    &-92.558&250.961&-92.558&250.961&-92.558&250.961&-92.558&250.961     \cr 
Constant        &-5.387 &-2.426     &-5.387 &-2.426 &-5.387 &-2.426 &-5.387 &-2.426 &-5.387 &-2.426      \cr   
Same religion   &-0.101 &0.889      &-0.131 &0.889  &-0.192 &0.889  &-0.282 &0.889  &-0.391 &0.889       \cr   
Same sex        &0.326  &1.455      &0.326  &1.455  &0.326  &1.455  &0.326  &1.455  &0.326  &1.455       \cr   
Same caste      &-0.248 &0.700      &-0.297 &0.700  &-0.443 &0.700  &-0.478 &0.700  &-0.711 &0.700       \cr   
Same language   &-0.121 &0.316      &-0.121 &0.316  &-0.121 &0.316  &-0.121 &0.316  &-0.121 &0.316       \cr   
Same family     &1.308  &3.016      &1.308  &3.055  &1.308  &3.463  &1.308  &3.653  &1.308  &4.189       \cr \hline  
\end{tabular}

\label{table:unionmixed01}
\caption*{\footnotesize{Note: $\hat{\mathcal{C}}_n(\alpha)$ is computed in Columns (1)-(5) as $\cup_{r_0\leq \bar{\mathcal{R}}_0, r_1\leq \bar{\mathcal{R}}_1} \mathcal{C}_n(\alpha; r_0, r_1)$, with $\bar{\mathcal{R}}_0=0.02$ and $\bar{\mathcal{R}}_1\in \{0.05, 0.1, 0.2, 0.3, 0.4\}$.}}
\end{sidewaystable}

\begin{table}[ht]       
   \centering
   \caption{Ratio of lengths of 95\% confidence intervals, $|\hat{\mathcal{C}}_n(\alpha)|/|\mathcal{C}_n(\alpha, 0,0)|$.}
      \begin{tabular}{l c c c c c}
\hline
& \shortstack{$r_0\leq 0.02$ \\ $r_1\leq 0.05$} & \shortstack{$r_0\leq 0.02$ \\ $r_1\leq 0.1$} & \shortstack{$r_0\leq 0.02$ \\ $r_1\leq 0.2$} & \shortstack{$r_0\leq 0.02$ \\ $r_1\leq 0.3$} & \shortstack{$r_0\leq 0.02$ \\ $r_1\leq 0.4$} \\
&  (1)                         & (2)                        & (3)                         & (4)                          & (5)                         \\ \hline
Reciprocation    &3.116&3.446&3.446&3.446&3.446\\
In degree        &8.084&8.174&8.174&8.174&8.174\\
Supported trust  &6.765&6.765&6.765&6.765&6.765\\
Constant         &8.211&8.211&8.211&8.211&8.211\\
Same religion    &6.868&7.078&7.498&8.126&8.880\\
Same sex         &8.062&8.062&8.062&8.062&8.062\\
Same caste       &8.328&8.755&10.035&10.348&12.394\\
Same language    &5.124&5.124&5.124&5.124&5.124\\
Same family      &7.452&7.623&9.406&10.233&12.570\\\hline 
\end{tabular}

      \label{table:ratiolengthmixed01}
\end{table}

\clearpage
\newpage
\appendix
\section{Proofs}\label{sec_a1}

\subsection{Proof of Lemmas in Section \ref{sec2}}

\begin{proof}[Proof of Lemma \ref{lemma1}]
By Assumption \ref{Ass2}, 
$$
G_{i,n}^\ast 
= 
\argmax_{g_{i,n}^\ast \in \mathcal{G}^n}E\left[U_{i}(g_{i,n}^\ast, G_{-i,n}^\ast, X, \varepsilon_{i}) \mid X, \varepsilon_{i}, \sigma_n \right] 
= 
\argmax_{g_{i,n}^\ast \in \mathcal{G}^n} \frac{1}{n} \sum_{j=1}^{n} g_{ij,n}^\ast
\left[(Z_{ij,n}^\ast)'\beta_0 + \varepsilon_{ij} \right]. 
$$
Therefore, $G_{ij,n}^\ast =  1 \left\{(Z_{ij,n}^\ast)'\beta_0+ \varepsilon_{ij}\geq 0\right\}$.
\end{proof}

\begin{proof}[Proof of Lemma \ref{belief_lemma}]
Define 
$$
D(r_0,r_1)=
\left(\begin{array}{cccc}
1-r_0-r_1&0&0&0\\
0&1-r_0-r_1&0&0\\
0&0&(1-r_0-r_1)^2&r_0(1-r_0-r_1)\\
0&0&0&1-r_0-r_1\\
\end{array}\right).
$$
By Assumption \ref{Ass3}, we can derive 
\begin{eqnarray*}
  && 
 E\left[G_{ki,n} G_{kj,n}\mid X, \sigma_n\right] 
 \\
 &=& 
E[Pr(G_{ki,n}=G_{kj,n}=1\mid G_{ki,n}^\ast ,G_{kj,n}^\ast,X, \sigma_n)\mid X, \sigma_n]\\
 &=& 
E[Pr(G_{ki,n}=1\mid G_{ki,n}^\ast ,G_{kj,n}^\ast,X, \sigma_n)Pr(G_{kj,n}=1\mid G_{ki,n}^\ast ,G_{kj,n}^\ast,X, \sigma_n)\mid X, \sigma_n]\\
 &=& 
\rho_0^2Pr(G_{ki,n}^\ast=0,G_{kj,n}^\ast=0\mid X, \sigma_n)
+\rho_0(1-\rho_1)Pr(G_{ki,n}^\ast=0,G_{kj,n}^\ast=1\mid X, \sigma_n)
\\&&+\rho_0(1-\rho_1)Pr(G_{ki,n}^\ast=1,G_{kj,n}^\ast=0\mid X, \sigma_n)
+(1-\rho_1)^2Pr(G_{ki,n}^\ast=1,G_{kj,n}^\ast=1\mid X, \sigma_n)
\\
 &=& 
 \rho_0^2+(1-\rho_0-\rho_1)^2 E\left[G_{ki,n}^\ast G_{kj,n}^\ast\mid X, \sigma_n\right]+\rho_0(1-\rho_0-\rho_1) E\left[G_{ki,n}^\ast+G_{kj,n}^\ast\mid X, \sigma_n\right],
 \end{eqnarray*}
 where the first equality is the law of iterated expectations, the second equality follows from Assumption \ref{Ass3} (i), the third equality follows from Assumption \ref{Ass3} (ii), and the last equality follows from $Pr(G_{ki,n}^\ast=0,G_{kj,n}^\ast=0\mid X, \sigma_n)=1-E\left[G_{ki,n}^\ast+G_{kj,n}^\ast\mid X, \sigma_n\right]+E[G_{ki,n}^\ast G_{kj,n}^\ast\mid X, \sigma_n]$, 
$Pr(G_{ki,n}^\ast=0,G_{kj,n}^\ast=1\mid X, \sigma_n)=E\left[G_{kj,n}^\ast\mid X, \sigma_n\right]-E[G_{ki,n}^\ast G_{kj,n}^\ast\mid X, \sigma_n]$, 
$Pr(G_{ki,n}^\ast=0,G_{kj,n}^\ast=1\mid X, \sigma_n)=E\left[G_{ki,n}^\ast\mid X, \sigma_n\right]-E[G_{ki,n}^\ast G_{kj,n}^\ast\mid X, \sigma_n]$, and 
$Pr(G_{ki,n}^\ast=1,G_{kj,n}^\ast=1\mid X, \sigma_n)=E[G_{ki,n}^\ast G_{kj,n}^\ast\mid X, \sigma_n]$.  
It follows then 
\begin{eqnarray*}
\gamma_{ij,n}
&=&
\left(\begin{array}{c}
E\left[G_{ji,n}\mid X, \sigma_n\right]\\
\frac{1}{n}\sum_{k}E\left[G_{kj,n}\mid X, \sigma_n\right]\\
\frac{1}{n}\sum_{k}E\left[G_{ki,n} G_{kj,n}\mid X, \sigma_n\right]\\
\frac{1}{n}\sum_{k}E\left[G_{ki,n}+G_{kj,n}\mid X, \sigma_n\right]
\end{array}\right)
\\
&=&
\left(\begin{array}{c}
\rho_0\\
\rho_0\\
\rho_0^2\\
\rho_0
\end{array}\right)
+
D(\rho_0,\rho_1)
\left(\begin{array}{c}
E\left[G_{ji,n}^\ast\mid X, \sigma_n\right]\\
\frac{1}{n}\sum_{k}E\left[G_{kj,n}^\ast\mid X, \sigma_n\right]\\
\frac{1}{n}\sum_{k}E\left[G_{ki,n}^\ast G_{kj,n}^\ast\mid X, \sigma_n\right]\\
\frac{1}{n}\sum_{k}E\left[G_{ki,n}^\ast+G_{kj,n}^\ast\mid X, \sigma_n\right]
\end{array}\right).
\end{eqnarray*}
Since $D(\rho_0,\rho_1)$ is invertible given $1-\rho_0-\rho_1\ne 0$, it follows that  
$$
\left(\begin{array}{c}
E\left[G_{ji,n}^\ast\mid X\right]\\
\frac{1}{n}\sum_{k}E\left[G_{kj,n}^\ast\mid X, \sigma_n\right]\\
\frac{1}{n}\sum_{k}E\left[G_{ki,n}^\ast G_{kj,n}^\ast\mid X, \sigma_n\right]\\
\frac{1}{n}\sum_{k}E\left[G_{ki,n}^\ast+G_{kj,n}^\ast\mid X, \sigma_n \right]
\end{array}\right)
=
D(\rho_0,\rho_1)^{-1}
\left(
\gamma_{ij,n}-
\left(\begin{array}{c}
\rho_0\\
\rho_0\\
\rho_0^2\\
\rho_0
\end{array}\right)
\right).
$$
The first three components of the right-hand side of the above equation are $\gamma_{ij,n}^\ast$, so 
$$
\gamma_{ij,n}^\ast
=
\left(\begin{array}{cccc}
1&0&0&0\\
0&1&0&0\\
0&0&1&0\\
\end{array}\right)
D(\rho_0,\rho_1)^{-1}
\left(
\gamma_{ij,n}-
\left(\begin{array}{c}
\rho_0\\
\rho_0\\
\rho_0^2\\
\rho_0
\end{array}\right)
\right)
=
c(\rho_0,\rho_1)+C(\rho_0,\rho_1)\gamma_{ij,n}.
$$
\end{proof}

\begin{proof}[Proof of Lemma \ref{Ass3_lemma}]
It suffices to show that $Pr(G_{ij,n}=1\mid X_{ij},\gamma_{ij,n}, \gamma_{ij,n}^\ast,X,\sigma_n)=\rho_0Pr(G_{ij,n}^\ast=0\mid X_{ij},\gamma_{ij,n}^\ast)+(1-\rho_1)Pr(G_{ij,n}^\ast=1\mid X_{ij},\gamma_{ij,n}^\ast)$.
Since $(X_{ij},\gamma_{ij,n}, \gamma_{ij,n}^\ast)$ are a function of $(X,\sigma_n)$, it follows that 
$$
Pr(G_{ij,n}=1\mid X_{ij},\gamma_{ij,n}, \gamma_{ij,n}^\ast,X,\sigma_n)
=
Pr(G_{ij,n}=1\mid X,\sigma_n).
$$
Using Assumptions \ref{Ass1}-\ref{Ass3}, 
\begin{eqnarray*}
Pr(G_{ij,n}=1\mid X,\sigma_n)
&=&
\rho_0Pr(G_{ij,n}^\ast=0\mid X,\sigma_n)+(1-\rho_1)Pr(G_{ij,n}^\ast=1\mid X,\sigma_n)\\
&=&
\rho_0Pr((Z_{ij,n}^\ast)'b + \varepsilon_{ij}<0\mid X,\sigma_n)+(1-\rho_1)Pr((Z_{ij,n}^\ast)'b + \varepsilon_{ij}\geq 0\mid X,\sigma_n)\\
&=&
\rho_0Pr((Z_{ij,n}^\ast)'b + \varepsilon_{ij}<0\mid Z_{ij,n}^\ast)+(1-\rho_1)Pr((Z_{ij,n}^\ast)'b + \varepsilon_{ij}\geq 0\mid Z_{ij,n}^\ast),
\end{eqnarray*}
where the first equality follows from Assumption \ref{Ass3}, the second follows from Lemma \ref{lemma1}, and the last follows from the independence between $\varepsilon$ and $X$.
\end{proof}

\subsection{Proof of Theorem \ref{theorem_identification}}

\begin{proof}
To show that every element $\theta$ of  $\Theta_I(P)$ satisfies Eq. \eqref{moment_TheoremID}, we can derive the following equalities: 
\begin{eqnarray*}
P(G_{ij,n}=1\mid X_{ij},\gamma_{ij,n})
&=&
P^\ast(G_{ij,n}=1\mid X_{ij},\gamma_{ij,n}) \\
&=&
P^\ast(G_{ij,n}=1\mid X_{ij}, \gamma_{ij,n}, \gamma_{ij,n}^\ast)\\
&=&
r_0+(1-r_0-r_1)P^\ast(G_{ij,n}^\ast=1\mid X_{ij},  \gamma_{ij,n}^\ast)\\
&=&
r_0+(1-r_0-r_1)P^\ast((Z_{ij,n}^\ast)'b + \varepsilon_{ij}\geq 0\mid X_{ij}, \gamma_{ij,n}^\ast)\\
&=&
r_0+(1-r_0-r_1)\Phi((\gamma_{ij,n}^\ast)'b_1+X_{ij}'b_2)\\
&=&
r_0+(1-r_0-r_1)\Phi((c(r_0,r_1)+C(r_0,r_1)\gamma_{ij,n})'b_1+X_{ij}'b_2),
\end{eqnarray*}
where the first equality follows from $P=P^\ast$ for the observables $(G_{ij,n},X_{ij},\gamma_{ij,n})$, the second equality follows because $\gamma_{ij,n}^\ast$ is a function of $ \gamma_{ij,n}$ in Condition \ref{misclas_prop}(ii), the third equality follows from Condition \ref{misclas_prop}(i), the fourth equality follows from Condition \ref{linear_index}, the fifth equality follows from Condition \ref{independence_assn}, and the last equality follows from Condition \ref{misclas_prop}(ii).
The rest of the proof is going to show that every element $\theta$ of  $\Theta$ satisfying Eq. \eqref{moment_TheoremID} belongs to $\Theta_I(P)$. 

Define the joint distribution $P^\ast$ in the following way. 
The marginal distribution of $\varepsilon_{ij}$ is standard normal.
The conditional distribution of $(\gamma_{ij,n},\gamma_{ij,n}^\ast,X_{ij})$ given $\varepsilon_{ij}$ is 
\begin{equation}\label{A_equ}
P^\ast((\gamma_{ij,n},\gamma_{ij,n}^\ast,X_{ij})\in B \mid \varepsilon_{ij})=P((\gamma_{ij,n},c(r_0,r_1)+C(r_0,r_1)\gamma_{ij,n},X_{ij})\in B)
\end{equation}
for all the measurable sets $B$. 
The conditional distribution of $G_{ij,n}^\ast$ given $(\gamma_{ij,n},\gamma_{ij,n}^\ast,X_{ij},\varepsilon_{ij})$ is  
\begin{equation}\label{B_equ}
P^\ast(G_{ij,n}^\ast=1\mid \gamma_{ij,n},\gamma_{ij,n}^\ast,X_{ij},\varepsilon_{ij})=1\{(Z_{ij,n}^\ast)'b + \varepsilon_{ij}\geq 0\}.
\end{equation}
The conditional distribution of $G_{ij,n}$ given $(G_{ij,n}^\ast,\gamma_{ij,n},\gamma_{ij,n}^\ast,X_{ij},\varepsilon_{ij})$  is  
\begin{equation}\label{C_equ}
P^\ast(G_{ij,n}=1\mid G_{ij,n}^\ast,\gamma_{ij,n},\gamma_{ij,n}^\ast,X_{ij},\varepsilon_{ij})
=
\begin{cases}
r_0&\mbox{ if }G_{ij,n}^\ast=0\\
1-r_1&\mbox{ if }G_{ij,n}^\ast=1.
\end{cases}
\end{equation}
Also it implies 
\begin{equation}\label{C_equ_weak}
P^\ast(G_{ij,n}=1\mid G_{ij,n}^\ast,Z_{ij},\gamma_{ij,n}^\ast)
=
\begin{cases}
r_0&\mbox{ if }G_{ij,n}^\ast=0\\
1-r_1&\mbox{ if }G_{ij,n}^\ast=1.
\end{cases}
\end{equation}

Note that $(P^\ast,\theta)$ satisfies Conditions 1-3, because Condition 1(i) follows because $\varepsilon_{ij}$ is normally distributed under $P^\ast$, Condition 1(ii) follows from Eq. (\ref{A_equ}). Condition 2 follows from Eq. (\ref{B_equ}). Condition 3(i) follows from Eq. (\ref{B_equ}) and (\ref{C_equ}), and Condition 3(ii) follows from Eq. (\ref{A_equ}). 

The distribution of $(G_{ij,n},X_{ij},\gamma_{ij,n})$ induced from $P^\ast$ is equal to $P$.
The distribution of $(X_{ij},\gamma_{ij,n})$ induced from $P^\ast$ is equal to that from $P$, by the construction of $P^\ast((\gamma_{ij,n},\gamma_{ij,n}^\ast,X_{ij})\in B \mid \varepsilon_{ij})$.
The equality of $P^\ast(G_{ij,n}=1\mid Z_{ij,n})=P(G_{ij,n}=1\mid Z_{ij,n})$ a.s. under $P^\ast$ is shown as follows.
Note that  
\begin{equation}\label{D_equ}
\gamma_{ij,n}^\ast=c(r_0,r_1)+C(r_0,r_1)\gamma_{ij,n}\mbox{ a.s. under }P^\ast.
\end{equation}
Then  
\begin{eqnarray*}
P^\ast(G_{ij,n}=1\mid Z_{ij,n})
&=&
P^\ast(G_{ij,n}=1\mid Z_{ij,n},\gamma_{ij,n}^\ast)\\
&=&
r_0P^\ast(G_{ij,n}^\ast=0\mid Z_{ij,n},\gamma_{ij,n}^\ast)+(1-r_1)P^\ast(G_{ij,n}^\ast=1\mid Z_{ij,n},\gamma_{ij,n}^\ast)\\
&=&
r_0+(1-r_0-r_1)P^\ast(G_{ij,n}^\ast=1\mid Z_{ij,n},\gamma_{ij,n}^\ast)\\
&=&
r_0+(1-r_0-r_1)E_{P^\ast}[P^\ast(G_{ij,n}^\ast=1\mid Z_{ij,n},\gamma_{ij,n}^\ast,\varepsilon_{ij})\mid Z_{ij,n},\gamma_{ij,n}^\ast]\\
&=&
r_0+(1-r_0-r_1)P^\ast((Z_{ij,n}^\ast)'b + \varepsilon_{ij}\geq 0\mid Z_{ij,n},\gamma_{ij,n}^\ast)\\
&=&
r_0+(1-r_0-r_1)\Phi((Z_{ij,n}^\ast)'b)\\
&=&
r_0+(1-r_0-r_1)\Phi((c(r_0,r_1)+C(r_0,r_1)\gamma_{ij,n})'b_1+X_{ij}'b_2)\\
&=&
P(G_{ij,n}=1\mid Z_{ij,n}),
\end{eqnarray*}
where the first and seventh equalities follow from Eq. (\ref{D_equ}), the second follows from Eq. (\ref{C_equ_weak}), the fifth follows from Eq. (\ref{B_equ}), and the last follows from Eq. (\ref{moment_TheoremID}).
\end{proof}

\subsection{Proof of Theorem \ref{theorem_inference}}
Theorem \ref{theorem_inference} follows because Lemma \ref{clt_m} and \ref{var_conv} in this appendix imply that, conditional on $(X,\sigma_n)$, $n\hat{m}_n(\theta)'\hat{S}(\theta)^{-1}\hat{m}_n(\theta)$ converges in distribution to the $\chi^2_J$ distribution. 

In the proof of this theorem, all the statements are conditional on $(X,\sigma_n)$. 
For any vector, the norm is understood as the Euclidean norm, and for any matrix the norm is induced by the Euclidean norm.   
Define  
\begin{eqnarray*}
u_{ij}(\theta)
&=&
(c(r_0,r_1)+C(r_0,r_1)\gamma_{ij,n})'b_1+X_{ij}'b_2\\
\hat{u}_{ij}(\theta)
&=&
(c(r_0,r_1)+C(r_0,r_1)\hat\gamma_{ij,n})'b_1+X_{ij}'b_2, 
\end{eqnarray*}
where we use $b=(b_1, b_2)\in \mathcal{B}$, and so $b_1$ represents the first three components of $b$ associated with the network externalities and $b_2$ represents the remaining components in $b$ associated with the homophily covariates.
For a generic random variable RV, define 
\begin{align*}
RV^\dagger = RV - E[RV\mid X, \sigma_n], 
\end{align*}
and note that $E[RV^\dagger\mid X, \sigma_n]=0$. For example, $G_{kj,n}^\dagger=G_{kj,n}-E\left[G_{kj,n}\mid X,\sigma_n\right]$.
Define  
\begin{eqnarray*}
\tilde\psi_k(\theta_0)
&=&
\frac{1}{n}\sum_{j\ne k}G_{kj,n}\zeta_{kj}-(1-\rho_0-\rho_1)\frac{1}{n^2}\sum_{i,j}\left(\phi(u_{ij}(\theta_0))\beta_1'C(\rho_0,\rho_1)\hat\psi_{\gamma,k,n}(X_{ij})\right)\zeta_{ij}.
\end{eqnarray*}

\begin{lemma}\label{gammastar_symmetric}
\begin{equation}\label{gammastar_symmetry}
1\{X_{i_{1},j_1}=X_{ij}\}
\left(
\begin{array}{c}
E[G_{j_1i_{1},n}^\ast\mid X, \sigma_{n}]-E[G_{ji,n}^\ast\mid X, \sigma_{n}]\\
\frac{1}{n}\sum_{k} (E[G_{kj_1,n}^\ast\mid X, \sigma_{n}] - E[G_{kj,n}^\ast\mid X, \sigma_{n}])\\
\frac{1}{n}\sum_{k} (E[G_{ki_{1},n}^\ast G_{kj_1,n}^\ast\mid X, \sigma_{n}] - E[G_{ki,n}^\ast G_{kj,n}^\ast\mid X, \sigma_{n}])
\\
\frac{1}{n}\sum_{k} \left(E\left[G_{ki_1,n}^\ast+G_{kj_1,n}^\ast\mid X, \sigma_n\right] - E\left[G_{ki,n}^\ast+G_{kj,n}^\ast\mid X, \sigma_n\right]\right)  
\end{array}
\right)
=0.
\end{equation}
\end{lemma}
\begin{proof}
This result follows from symmetry of the equilibrium and it is shown in a similar way to Lemma 1 in \citet{leung2015two}.
\end{proof}

\begin{lemma}\label{bounds_many}
$$
\max\{\|\hat\psi_{\gamma,k, n}(X_{ij})\|,\|\psi_{\gamma,k, n}(X_{ij})\|\}\leq \frac{4}{\min_{x}\hat{p}(x)}
$$
$$
\max_{i}\{\|\tilde\psi_i(\theta_0)\|,\|\hat\psi_i(\theta_0)\|,\|\psi_i(\theta_0)\|\}\leq  1+(1-\rho_0-\rho_1)\phi(0) \|\beta_1'C(\rho_0,\rho_1)\|\frac{4}{\min_{x}\hat{p}(x)}.
$$
\end{lemma}
\begin{proof}
The bound for $\|\hat\psi_{\gamma,k, n}(X_{ij})\|$ is derived from 
\begin{eqnarray*}
\|\hat\psi_{\gamma,k, n}(x)\|
&\leq&
\frac{1}{n^2} \sum_{i_{1},j_1} 
\frac{1\{X_{i_{1},j_1}=x\}}{\hat{p}(x)} 
\left\|
\left(  
\begin{array}{c}
0\\
G_{kj_1}\\
G_{ki_{1}}G_{kj_1}\\
G_{ki_{1}}+G_{kj_1}
\end{array}
\right)
\right\|
+
\frac{1}{n} \sum_{i_{1}} 
\frac{1\{X_{i_{1},k}=x\}}{\hat{p}(x)} 
\left\|
\left(  
\begin{array}{c}
G_{ki_{1}}\\
0\\
0\\
0
\end{array}
\right)\right\|
\\
&\leq&  
\frac{\sqrt{6}}{\min_x\hat{p}(x)}+\frac{1}{\min_x\hat{p}(x)}\\
\\
&\leq&  
\frac{4}{\min_x\hat{p}(x)}.
\end{eqnarray*}
The bound for $\|\psi_{\gamma,k, n}(X_{ij})\|$ is similarly derived. 

The bound for $\|\tilde\psi_i(\theta)\|$ is derived from 
\begin{eqnarray*}
\|\tilde\psi_i(\theta)\|
&\leq&
\max_{j\ne i}\left|G_{ij,n}\right|+(1-\rho_0-\rho_1)\max_{l,j}\left|\phi(u_{lj}(\theta_0))\beta_1'C(\rho_0,\rho_1)\hat\psi_{\gamma,i, n}(X_{lj})\right|\\
 &\leq&
1+(1-\rho_0-\rho_1)\phi(0) \|\beta_1'C(\rho_0,\rho_1)\|\frac{4}{\min_{x}\hat{p}(x)}.
\end{eqnarray*}
The bound for $\|\hat\psi_i(\theta)\|$ is similarly derived.

The bound for $\|\psi_i(\theta)\|$ is derived from  
\begin{eqnarray*}
\|\psi_i(\theta_0)\|
&\leq&
\max_{j\ne i}\left|G_{ij,n}-\rho_0-(1-\rho_0-\rho_1)\Phi(u_{ij}(\theta_0))\right|\\&&+(1-\rho_0-\rho_1)\max_{l,j}\|\phi(u_{ij}(\theta_0))\beta_1'C(\rho_0,\rho_1)\| \|\psi_{\gamma,i, n}(X_{lj})\|
\\
&\leq&
1+(1-\rho_0-\rho_1)\phi(0)\|\beta_1'C(\rho_0,\rho_1)\|\frac{4}{\min_{x}\hat{p}(x)}.
\end{eqnarray*}
\end{proof}

\begin{lemma}\label{gamma_influe}
$$
\hat\gamma_{ij}-\gamma_{ij,n}=\frac{1}{n}\sum_{k} \psi_{\gamma,k, n}(X_{ij})
$$ 
and 
$$
\sup_{i,j}\left\|\hat\gamma_{ij}-\gamma_{ij,n}\right\|=O_p(n^{-1/2})\mbox{ given }(X,\sigma_n).
$$ 
\end{lemma}
 \begin{proof}
First, from Lemma \ref{gammastar_symmetric} and Assumption \ref{Ass3}, we can derive  
\begin{equation}\label{gamma_symmetry}
1\{X_{i_{1},j_1}=X_{ij}\}
\left(
\begin{array}{c}
E[G_{j_1i_{1},n}\mid X, \sigma_{n}]-E[G_{ji,n}\mid X, \sigma_{n}]\\
\frac{1}{n}\sum_{k} (E[G_{kj_1,n}\mid X, \sigma_{n}] - E[G_{kj,n}\mid X, \sigma_{n}])\\
\frac{1}{n}\sum_{k} (E[G_{ki_{1},n} G_{kj_1,n}\mid X, \sigma_{n}] - E[G_{ki,n} G_{kj,n}\mid X, \sigma_{n}])\\
\frac{1}{n}\sum_{k} (E[(G_{ki_{1},n}+G_{kj_1,n})\mid X, \sigma_{n}] - E[(G_{ki,n}+G_{kj,n})\mid X, \sigma_{n}])
\end{array}
\right)
=0.
\end{equation} 

Using Eq. \eqref{gamma_symmetry}, we have  
\begin{eqnarray*}
&&
\hat\gamma_{ij}-\gamma_{ij,n} 
\\
&=&
\frac{1}{n^2}\sum_{i_{1},j_1}\frac{1\{X_{i_{1},j_1}=X_{ij}\}}{\frac{1}{n^2}\sum_{i_{1},j_1}1\{X_{i_{1},j_1}=X_{ij}\}}
\left(
\begin{array}{c}
G_{j_1i_{1},n}-E[G_{ji,n}\mid X, \sigma_{n}]\\
\frac{1}{n}\sum_{k} (G_{kj_1,n} - E[G_{kj,n}\mid X, \sigma_{n}])\\
\frac{1}{n}\sum_{k} (G_{ki_{1},n} G_{kj_1,n} - E[G_{ki,n} G_{kj,n}\mid X, \sigma_{n}])\\
\frac{1}{n}\sum_{k} ((G_{ki_{1},n}+G_{kj_1,n}) - E[(G_{ki,n}+G_{kj,n})\mid X, \sigma_{n}])
\end{array}
\right)
\\
&=&
\frac{1}{n^2}\sum_{i_{1},j_1}\frac{1\{X_{i_{1},j_1}=X_{ij}\}}{\frac{1}{n^2}\sum_{i_{1},j_1}1\{X_{i_{1},j_1}=X_{ij}\}}
\left(
\begin{array}{c}
G_{j_1i_{1},n}^\dagger\\
\frac{1}{n}\sum_{k} G_{kj_1,n}^\dagger \\
\frac{1}{n}\sum_{k} (G_{ki_{1},n} G_{kj_1,n})^\dagger\\
\frac{1}{n}\sum_{k} ((G_{ki_{1},n}+G_{kj_1,n}))^\dagger
\end{array}
\right)
\\
&=&
\frac{1}{n}\sum_{k} \psi_{\gamma,k, n}(X_{ij}).
\end{eqnarray*}

Since $X_{ij}$ has a finite support, the uniform convergence over $i,j$ follows from the point convergence for every $i,j$.
By Lyapunov's central limit theorem, it suffices to show 
that $E[\psi_{\gamma,k, n}(X_{ij})\mid X,\sigma_n]=0$ and that $\psi_{\gamma,k, n}(X_{ij})$ is independent across $k$ given $(X,\sigma_n)$. 
The equality $E[\psi_{\gamma,k, n}(X_{ij})\mid X,\sigma_n]=0$ follows from 
\begin{align*}
E[\psi_{\gamma,k, n}(X_{ij})\mid X,\sigma_n]
&=
\frac{1}{n^2} \sum_{i_{1},j_1} 
\left(\frac{1\{X_{i_{1},j_1}=X_{ij}\}}{\hat{p}(X_{ij})} \right) 
\left(  
\begin{array}{c}
0\\
E\left[ G_{kj_1,n}^{\dagger}  \mid X,\sigma_n \right]\\
E\left[(G_{ki_{1},n}G_{kj_1,n})^{\dagger} \mid X,\sigma_n \right]\\
E\left[(G_{ki_{1},n}+G_{kj_1,n})^{\dagger} \mid X,\sigma_n \right]
\end{array}
\right)
\\
&+
\frac{1}{n} \sum_{i_{1}} 
\left(\frac{1\{X_{i_{1},k}=X_{ij}\}}{\hat{p}(X_{ij})} \right) 
\left(  
\begin{array}{c}
E\left[ G_{ki_{1},n}^{\dagger} \mid X,\sigma_n \right]\\
0\\
0\\
0
\end{array}
\right)
\\
& =0
\end{align*}
since $E\left[ RV^{\dagger} \mid X,\sigma_n \right] =0$ by definition of $RV^\dagger$. The conditional independence of $\psi_{\gamma,k, n}(X_{ij})$ across $k$ is shown as follows.
Note that $\psi_{\gamma,k, n}(X_{ij})$ does not depend on $G_{-k,n}$, so it is a function of $\varepsilon_k$, $(X,\sigma_n)$. Therefore, it follows from Assumptions \ref{Ass1} that $\psi_{\gamma,k, n}(X_{ij})$ is independent across $k$ given $(X,\sigma_n)$.
\end{proof}

\begin{lemma}\label{diff_psi}
$\max_{i}\|\hat\psi_i(\theta_0)-\tilde\psi_i(\theta_0)\|=o_p(1)$ given $(X,\sigma_n)$.
\end{lemma}
\begin{proof}
Note that 
\begin{eqnarray*}
\hat\psi_i(\theta_0)-\tilde\psi_i(\theta_0)
&=&
-(1-\rho_0-\rho_1)\frac{1}{n^2}\sum_{l,j}\left(\phi(\hat{u}_{lj}(\theta_0))-\phi(u_{lj}(\theta_0)) \right)\beta_1'C(\rho_0,\rho_1) \hat\psi_{\gamma,i, n}(X_{lj}) \zeta_{lj}.
\end{eqnarray*}
Then 
\begin{eqnarray*}
\|\hat\psi_i(\theta_0)-\tilde\psi_i(\theta_0)\|
&\leq&
\|\beta_1'C(\rho_0,\rho_1)\|\max_{l,j}\left|\phi(\hat{u}_{lj}(\theta_0))-\phi(u_{lj}(\theta_0))\right| \|\hat\psi_{\gamma,i, n}(X_{lj})\|\\
&\leq&
\phi(0)\|\beta_1'C(\rho_0,\rho_1)\|\max_{l,j}\max\{|\hat{u}_{lj}(\theta_0)|,|u_{lj}(\theta_0)|\}|\hat{u}_{lj}(\theta_0)-u_{lj}(\theta_0)|\|\hat\psi_{\gamma,i, n}(X_{lj})\|,
\end{eqnarray*}
where the last inequality follows from the mean value expansion of the normal pdf $\phi$:  
$|\phi(u_1)-\phi(u_2)|\leq\max_{u_1\leq u\leq u_2}|\phi'(u)||u_1-u_2|\leq\phi(0)\max\{|u_1|,|u_2|\}|u_1-u_2|$. 
Since 
\begin{eqnarray*}
|u_{lj}(\theta_0)|
&\leq&
(\|c(\rho_0,\rho_1)\|+\|C(\rho_0,\rho_1)\|\|\gamma_{lj,n})\|\|\beta_1\|+\|X_{lj}\|\|\beta_2\|\\
&\leq&
(\|c(\rho_0,\rho_1)\|+4\|C(\rho_0,\rho_1)\|)\|\beta_1\|+\max_{x}\|x\|\|\beta_2\|\\
|\hat{u}_{lj}(\theta_0)|
&\leq&
(\|c(\rho_0,\rho_1)\|+4\|C(\rho_0,\rho_1)\|)\|\beta_1\|+\max_{x}\|x\|\|\beta_2\|\\
|\hat{u}_{lj}(\theta_0)-u_{lj}(\theta_0)|
&=&
|C(\rho_0,\rho_1)(\hat\gamma_{lj}-\gamma_{lj,n})'\beta_1|\\
&\leq&
\|C(\rho_0,\rho_1)\|\|\beta_1\|\max_{lj}\|\hat\gamma_{lj}-\gamma_{lj,n}\|, 
\end{eqnarray*}
it follows that 
$$
\max_{i}\|\hat\psi_i(\theta_0)-\tilde\psi_i(\theta_0)\|=O_p(\max_{lj}\|\hat\gamma_{lj}-\gamma_{lj,n})\|)=o_p(1).
$$
\end{proof}

\begin{lemma}\label{psi_ind}
$\psi_i(\theta_0)$ is independent across $i$ given $(X,\sigma_n)$. 
\end{lemma}
\begin{proof}
$\psi_i(\theta_0)$ does not depend on $G_{-i,n}$, so it is a function of $(\varepsilon_i,X,\sigma_n)$.  By the independence of $\varepsilon_i$ across $i$, it implies the statement of this lemma. 
\end{proof}

\begin{lemma}
$\hat{m}_n(\theta_0)=\frac{1}{n}\sum_{i=1}^n\psi_i(\theta_0)+o_p(n^{-1/2})$ given $(X,\sigma_n)$.
\end{lemma}
\begin{proof}
Note that 
\begin{eqnarray*}
   &&
   \hat{m}_n(\theta_0)-\frac{1}{n}\sum_{i=1}^n\psi_i(\theta_0)
   \\
   &=&
   (1-\rho_0-\rho_1)\frac{1}{n^2}\sum_{i,j}\left(\Phi(\hat{u}_{ij}(\theta_0))-\Phi(u_{ij}(\theta_0))-\phi(u_{ij}(\theta_0))\beta_1'C(\rho_0,\rho_1)(\hat\gamma_{ij}-\gamma_{ij,n})\right)\zeta_{ij}   
\end{eqnarray*}
By the second-order Taylor expansion of the normal cdf $\Phi$, 
$$
\Phi(u_1)=\Phi(u_2)+\phi(u_2)(u_2-u_1)+R(u_1,u_2)
$$
where 
$$
|R_{ij}|\leq\frac{1}{2}\max_{u_1\leq u\leq u_2}\phi'(u)|u_1-u_2|^2\leq\frac{1}{2}\phi(0)\max\{|u_1|,|u_2|\}|u_1-u_2|^2.
$$
Since 
\begin{eqnarray*}
\max\{|u_{lj}(\theta_0)|,|\hat{u}_{lj}(\theta_0)|\}
&\leq&
(\|c(\rho_0,\rho_1)\|+4\|C(\rho_0,\rho_1)\|)\|\beta_1\|+\max_{x}\|x\|\|\beta_2\|\\
|\hat{u}_{lj}(\theta_0)-u_{lj}(\theta_0)|
&\leq&
\|C(\rho_0,\rho_1)\|\|\beta_1\|\max_{lj}\|\hat\gamma_{lj}-\gamma_{lj,n}\|, 
\end{eqnarray*}
it follows that 
$$
\|\hat{m}_n(\theta_0)-\frac{1}{n}\sum_{i=1}^n\psi_i(\theta_0)\|
=O_p(\max_{lj}\|\hat\gamma_{lj}-\gamma_{lj,n})\|^2)
=O_p(n^{-1}).
$$
\end{proof}

\begin{lemma}\label{clt_m}
Conditional on $(X,\sigma_n)$, 
$$
\hat{m}_n(\theta_0)=o_P(1)
$$
and 
$$
Var(\psi_i(\theta_0)\mid X,\sigma_n)^{-1/2}\sqrt{n}\hat{m}_n(\theta_0)\rightarrow_dN(0,I).
$$
\end{lemma}
\begin{proof}
By Lemmas \ref{bounds_many} and \ref{psi_ind} and Lyapunov's central limit theorem, it suffices to show $E[\psi_i(\theta_0)\mid X,\sigma_n]=0$.
We can derive 
\begin{eqnarray*}
E[\psi_i(\theta_0)\mid X,\sigma_n]
&=&
\frac{1}{n}\sum_{j\ne i}\left(E[G_{ij,n}\mid X,\sigma_n]-\rho_0-(1-\rho_0-\rho_1)\Phi(u_{ij}(\theta_0))\right)\zeta_{ij}
\\&&
-(1-\rho_0-\rho_1)\frac{1}{n^2}\sum_{l,j}\left(
\phi(u_{lj}(\theta_0))\beta_1'C(\rho_0,\rho_1)
 E[\psi_{\gamma,i, n}(X_{lj})\mid X,\sigma_n]
 \right)\zeta_{lj}\\ 
&=&
0, 
\end{eqnarray*}
because 
\begin{eqnarray*}
E[G_{ij,n}\mid X,\sigma_n]
&=&
\rho_0+(1-\rho_0-\rho_1)\Phi(u_{ij}(\theta_0))\\
E[\psi_{\gamma,i, n}(X_{lj})\mid X,\sigma_n]
&=&
\frac{1}{n^2} \sum_{i_{1},j_1} 
\left(\frac{1\{X_{i_{1},j_1}=X_{ij}\}}{\hat{p}(X_{ij})} \right) 
\left(  
\begin{array}{c}
0\\
E[G_{kj_1,n}^{\dagger}\mid X,\sigma_n]\\
E[(G_{ki_{1},n}G_{kj_1,n})^{\dagger}\mid X,\sigma_n]\\
E[(G_{ki_{1},n}+G_{kj_1,n})^{\dagger}\mid X,\sigma_n]
\end{array}
\right)
\\
&&
+
\frac{1}{n} \sum_{i_{1}} 
\left(\frac{1\{X_{i_{1},k}=X_{ij}\}}{\hat{p}(X_{ij})} \right) 
\left(  
\begin{array}{c}
E[G_{ki_{1},n}^{\dagger}\mid X,\sigma_n]\\
0\\
0\\
0
\end{array}
\right)\\
&=&
0.
\end{eqnarray*}
Note that $E[RV^{\dagger}\mid X,\sigma_n]=0$ by the definition of $RV^\dagger$.
\end{proof}

\begin{lemma}\label{var_conv}
$\hat{S}(\theta_0)=Var(\psi_i(\theta_0)\mid X,\sigma_n)+o_p(1)$ given $(X,\sigma_n)$. 
\end{lemma}
\begin{proof}
First, we are going to show that $\hat{S}(\theta_0)=\frac{1}{n}\sum_{i=1}^n\tilde\psi_i(\theta_0)\tilde\psi_i(\theta_0)'-
\left(\frac{1}{n}\sum_{i=1}^n\tilde\psi_i(\theta_0)\right)
\left(\frac{1}{n}\sum_{i=1}^n\tilde\psi_i(\theta_0)\right)'
+o_p(1)$.
Since 
\begin{eqnarray*}
&&
\hat{S}(\theta_0)-\frac{1}{n}\sum_{i=1}^n\tilde\psi_i(\theta_0)\tilde\psi_i(\theta_0)'+\left(\frac{1}{n}\sum_{i=1}^n\tilde\psi_i(\theta_0)\right)\left(\frac{1}{n}\sum_{i=1}^n\tilde\psi_i(\theta_0)\right)'
\\
&=&
\frac{1}{n}\sum_{i=1}^n(\hat\psi_i(\theta_0)-\tilde\psi_i(\theta_0))(\hat\psi_i(\theta_0)-\tilde\psi_i(\theta_0))'
\\&&+\frac{1}{n}\sum_{i=1}^n\tilde\psi_i(\theta_0)(\hat\psi_i(\theta_0)-\tilde\psi_i(\theta_0))'
\\&&+\frac{1}{n}\sum_{i=1}^n(\hat\psi_i(\theta_0)-\tilde\psi_i(\theta_0))\tilde\psi_i(\theta_0)'
\\&&-\left(\frac{1}{n}\sum_{i=1}^n(\hat\psi_i(\theta_0)-\tilde\psi_i(\theta_0))\right)\left(\frac{1}{n}\sum_{i=1}^n(\hat\psi_i(\theta_0))\right)'
\\&&-\left(\frac{1}{n}\sum_{i=1}^n\tilde\psi_i(\theta_0)\right)\left(\frac{1}{n}\sum_{i=1}^n(\hat\psi_i(\theta_0)-\tilde\psi_i(\theta_0))\right)',
\end{eqnarray*}
it follows that 
\begin{eqnarray*}
&&
\left\|\hat{S}(\theta_0)-\frac{1}{n}\sum_{i=1}^n\tilde\psi_i(\theta_0)\tilde\psi_i(\theta_0)'+\left(\frac{1}{n}\sum_{i=1}^n\tilde\psi_i(\theta_0)\right)\left(\frac{1}{n}\sum_{i=1}^n\tilde\psi_i(\theta_0)\right)'
\right\|\\
&&\quad\leq
\max_{i}\|\hat\psi_i(\theta_0)-\tilde\psi_i(\theta_0)\|^2+3\max_{i}\|\hat\psi_i(\theta_0)-\tilde\psi_i(\theta_0)\|\max_{i}\|\tilde\psi_i(\theta_0)\|+\max_{i}\|\hat\psi_i(\theta_0)-\tilde\psi_i(\theta_0)\|\max_{i}\|\hat\psi_i(\theta_0)\|.
\end{eqnarray*}
Thus it suffices to show that $\max_{i}\|\hat\psi_i(\theta_0)-\tilde\psi_i(\theta_0)\|=o_p(1)$ and $\max_{i}\{\|\tilde\psi_i(\theta_0)\|,\|\hat\psi_i(\theta_0)\|\}=O_p(1)$.
They are shown in Lemmas \ref{bounds_many} and \ref{diff_psi}. 

Second, we are going to show that $\hat{S}(\theta_0)=Var(\tilde\psi_i(\theta_0)\mid X,\sigma_n)+o_p(1)$.
It suffices to show $E[\|\tilde\psi_i(\theta_0)\|^4\mid X,\sigma_n]<\infty$. 
By the triangle inequality, 
\begin{eqnarray*}
E[\|\tilde\psi_i(\theta_0)\|^4\mid X,\sigma_n]^{1/4}
&\leq&
\frac{1}{n}\sum_{j\ne i}E\left[\left\|G_{ij,n}\right\|^4\mid X,\sigma_n\right]^{1/4}\\&&+\frac{1}{n^2}\sum_{l,j}E\left[\left\|\phi(u_{ij}(\theta_0))\beta_1'C(\rho_0,\rho_1)\hat\psi_{\gamma,i, n}(X_{lj})\right\|^4\mid X,\sigma_n\right]^{1/4}\\
&\leq&
\frac{1}{n}\sum_{j\ne i}\left(E[\left\|G_{ij,n}\right\|^4\mid X,\sigma_n]^{1/4}\right)\\&&+\frac{1}{n^2}\sum_{l,j}\phi(u_{ij}(\theta_0))\beta_1'C(\rho_0,\rho_1)E\left[\left\|\hat\psi_{\gamma,i, n}(X_{lj})\right\|^4\mid X,\sigma_n\right]^{1/4}\\
&\leq&
1+\frac{1}{n^2}\sum_{l,j}\phi(u_{ij}(\theta_0))\beta_1'C(\rho_0,\rho_1)E\left[\left\|\hat\psi_{\gamma,i, n}(X_{lj})\right\|^4\mid X,\sigma_n\right]^{1/4}\\
&<&
\infty,
\end{eqnarray*}
where the last inequality follows from Lemma \ref{bounds_many}. 

Third, we are going to show that $Var(\tilde\psi_i(\theta_0)\mid X,\sigma_n)=Var(\psi_i(\theta_0)\mid X,\sigma_n)$.
Note that $\tilde\psi_i(\theta_0)-\psi_i(\theta_0)$ is a function of $(X,\sigma_n)$, so the conditional variances are the same. 
\end{proof}

\subsection{Proof of Theorem \ref{theorem_inference_2}}
As in the previous section, all the statements in this appendix are conditional on $(X,\sigma_n)$. 
Theorem \ref{theorem_inference_2} follows from Lemma \ref{lemma:clt_hatbeta}. 

\begin{lemma}\label{lemma:beta_identi_true_rho}
$\beta$ is the unique maximizer of $E\left[\boldsymbol{Q}_n(b)\mid X,\sigma_n\right]$, where 
$$
\boldsymbol{Q}_n(b)=\frac{1}{n^2}\sum_{i,j}\log\left(\Psi(b,\rho_0,\rho_1,X_{ij},\gamma_{ij,n})^{G_{ij,n}}(1-\Psi(b,\rho_0,\rho_1,X_{ij},\gamma_{ij,n}))^{1-G_{ij,n}}\right).
$$
\end{lemma}
\begin{proof}
Applying Jensen's inequality to the logarithm function, we have  
\begin{eqnarray*}
&&
E\left[\boldsymbol{Q}_n(b)\mid X,\sigma_n\right]-E\left[\boldsymbol{Q}_n(\beta)\mid X,\sigma_n\right]
\\
&=&
\frac{1}{n^2}\sum_{i,j}\left(\Psi(\theta_0,X_{ij},\gamma_{ij,n})\log \frac{\Psi(b,\rho_0,\rho_1,X_{ij},\gamma_{ij,n})}{\Psi(\theta_0,X_{ij},\gamma_{ij,n})}+(1-\Psi(\theta_0,X_{ij},\gamma_{ij,n}))\log\frac{1-\Psi(b,\rho_0,\rho_1,X_{ij},\gamma_{ij,n})}{1-\Psi(\theta_0,X_{ij},\gamma_{ij,n})}\right)\\
&\leq&
\log\left(\frac{1}{n^2}\sum_{i,j}\left(\Psi(\theta_0,X_{ij},\gamma_{ij,n}) \frac{\Psi(b,\rho_0,\rho_1,X_{ij},\gamma_{ij,n})}{\Psi(\theta_0,X_{ij},\gamma_{ij,n})}+(1-\Psi(\theta_0,X_{ij},\gamma_{ij,n}))\frac{1-\Psi(b,\rho_0,\rho_1,X_{ij},\gamma_{ij,n})}{1-\Psi(\theta_0,X_{ij},\gamma_{ij,n})}\right)\right)\\
&=&
0.
\end{eqnarray*}
It suffices to show that the equality holds only when $b=\beta$.  
By Jensen's inequality, the equality holds if and only if  
\begin{equation}\label{eq:proof_beta_identi_true_rho}
\frac{\Psi(b,\rho_0,\rho_1,X_{ij},\gamma_{ij,n})}{\Psi(\theta_0,X_{ij},\gamma_{ij,n})}=1\mbox{ for every }i,j.
\end{equation}
Eq. (\ref{eq:proof_beta_identi_true_rho}) implies $((\gamma_{ij,n}^\ast)',X_{ij}')\beta=((\gamma_{ij,n}^\ast)',X_{ij}')b\mbox{ for every }i,j$.
Since $\{((\gamma_{ij,n}^\ast)',X_{ij}')':i,j\}$ is not contained in any proper linear subspace of $\mathbb{R}^{d+3}$, we have $\beta=b$.  
\end{proof}

\begin{lemma}\label{lemma:unfi_conv}
Conditional on $(X,\sigma_n)$, 
$$
\sup_{b\in\mathcal{B}}|\boldsymbol{Q}_n(b)-E\left[\boldsymbol{Q}_n(b)\mid X,\sigma_n\right]|=o_p(1)
$$
$$
\sup_{b\in\mathcal{B}}\|\frac{\partial^2}{\partial b\partial b'}\boldsymbol{Q}_n(b)-E\left[\frac{\partial^2}{\partial b\partial b'}\boldsymbol{Q}_n(b)\mid X,\sigma_n\right]\|=o_p(1).
$$
\end{lemma}
\begin{proof}
They follow from \citet[Proposition 1]{jenish2009central} as in the proof of \citet[][Theorem 2]{leung2015two}.
\end{proof}

\begin{lemma}\label{lemma:QMLE_consisitent}
$\hat\beta(\rho_0,\rho_1)\rightarrow_{a.s.}\beta$.
\end{lemma}
\begin{proof}
By Lemma \ref{lemma:beta_identi_true_rho} and \citet[Theorem 3.3]{white1988unified}, it suffices to show that  
$$\sup_{b\in\mathcal{B}}|\hat{\boldsymbol{Q}}_n(b,\rho_0,\rho_1)-E\left[\boldsymbol{Q}_n(b)\mid X,\sigma_n\right]|\rightarrow_p0.$$
By Lemma \ref{lemma:unfi_conv}, we need to show that $\sup_{b\in\mathcal{B}}|\hat{\boldsymbol{Q}}_n(b,\rho_0,\rho_1)-\boldsymbol{Q}_n(b)|\rightarrow_p0$.
Some calculations yield  
{
\small
\begin{eqnarray*}
&&
|\hat{\boldsymbol{Q}}_n(b,\rho_0,\rho_1)-\boldsymbol{Q}_n(b)|
\\
&=&
\left|\frac{1}{n^2}\sum_{i,j}\log\left(\left(\frac{\Psi(b,\rho_0,\rho_1,X_{ij},\hat\gamma(X_{ij}))}{\Psi(b,\rho_0,\rho_1,X_{ij},\gamma_{ij,n})}\right)^{G_{ij,n}}\left(\frac{1-\Psi(b,\rho_0,\rho_1,X_{ij},\hat\gamma(X_{ij}))}{1-\Psi(b,\rho_0,\rho_1,X_{ij},\gamma_{ij,n})}\right)^{1-G_{ij,n}}\right)\right|\\
&\leq&
\max_{i,j}\max\left\{\left|\log\left(\frac{\Psi(b,\rho_0,\rho_1,X_{ij},\hat\gamma(X_{ij}))}{\Psi(b,\rho_0,\rho_1,X_{ij},\gamma_{ij,n})}\right)\right|,\left|\log\left(\frac{1-\Psi(b,\rho_0,\rho_1,X_{ij},\hat\gamma(X_{ij}))}{1-\Psi(b,\rho_0,\rho_1,X_{ij},\gamma_{ij,n})}\right)\right|\right\}\\
&\leq&
\max_{i,j}
\frac{|\Psi(b,\rho_0,\rho_1,X_{ij},\hat\gamma(X_{ij}))-\Psi(b,\rho_0,\rho_1,X_{ij},\gamma_{ij,n})|}{\min\{\Psi(b,\rho_0,\rho_1,X_{ij},\hat\gamma(X_{ij})),\Psi(b,\rho_0,\rho_1,X_{ij},\gamma_{ij,n}),1-\Psi(b,\rho_0,\rho_1,X_{ij},\hat\gamma(X_{ij})),1-\Psi(b,\rho_0,\rho_1,X_{ij},\gamma_{ij,n})\}}\\
&\leq&
\max_{i,j}
\frac{|\Psi(b,\rho_0,\rho_1,X_{ij},\hat\gamma(X_{ij}))-\Psi(b,\rho_0,\rho_1,X_{ij},\gamma_{ij,n})|}{\min\{\Psi(b,\rho_0,\rho_1,X_{ij},\gamma_{ij,n}),1-\Psi(b,\rho_0,\rho_1,X_{ij},\gamma_{ij,n})\}-|\Psi(b,\rho_0,\rho_1,X_{ij},\hat\gamma(X_{ij}))-\Psi(b,\rho_0,\rho_1,X_{ij},\gamma_{ij,n})|}\\
&\leq&
\frac{\mathrm{term1}}{\mathrm{term2}-\mathrm{term1}}
\end{eqnarray*}
}
where the second inequality follows from $|\log(x)|\leq\max\{|x-1|,|x-1|/x\}$ for $x>0$ and the last equation uses 
\begin{eqnarray*}
\mathrm{term1}&=&\max_{i,j}|\Psi(b,\rho_0,\rho_1,X_{ij},\hat\gamma(X_{ij}))-\Psi(b,\rho_0,\rho_1,X_{ij},\gamma_{ij,n})|
\\
\mathrm{term2}&=&\min_{i,j}\min\{\Psi(b,\rho_0,\rho_1,X_{ij},\gamma_{ij,n}),1-\Psi(b,\rho_0,\rho_1,X_{ij},\gamma_{ij,n})\}.
\end{eqnarray*}
Since $\min_{i,j}\min\{\Psi(b,\rho_0,\rho_1,X_{ij},\gamma_{ij,n}),1-\Psi(b,\rho_0,\rho_1,X_{ij},\gamma_{ij,n})\}$ is bounded away from zero (because the support of $X_{ij}$ is finite), 
the uniform convergence of $\hat{\boldsymbol{Q}}_n(b,\rho_0,\rho_1)-\boldsymbol{Q}_n(b)$ follows from 
\begin{eqnarray*}
&&
\sup_{b\in\mathcal{B}}\max_{i,j}|\Psi(b,\rho_0,\rho_1,X_{ij},\hat\gamma(X_{ij}))-\Psi(b,\rho_0,\rho_1,X_{ij},\gamma_{ij,n})|
\\
&&\qquad=
(1-\rho_0-\rho_1)\sup_{b\in\mathcal{B}}\max_{i,j}\left|\Phi(\hat{u}_{ij}(b,\rho_0,\rho_1))-\Phi(u_{ij}(b,\rho_0,\rho_1))\right|\\
&&\qquad=
O_p\left(\|\max_{ij}\|\hat\gamma_{ij}-\gamma_{ij,n}\|\right).
\end{eqnarray*}
\end{proof}

\begin{lemma}\label{lemma:min_eig_hess}
The minimum eigenvalue of $\{\left.E\left[\frac{\partial^2}{\partial b\partial b'}\boldsymbol{Q}_n(b)\mid X,\sigma_n\right]\right|_{b=\beta}\}$ is bounded away from zero.
\end{lemma}
\begin{proof}
We have the following equalities: 
\begin{eqnarray*}
&&
\left.E\left[\frac{\partial^2}{\partial b\partial b'}\boldsymbol{Q}_n(b)\mid X,\sigma_n\right]\right|_{b=\beta}
\\
&=&
\frac{1}{n^2}\sum_{i,j}\frac{\left.\frac{\partial}{\partial b}\Psi(b,\rho_0,\rho_1,X_{ij},\gamma_{ij,n})\right|_{b=\beta}\left.\frac{\partial}{\partial b'}\Psi(b,\rho_0,\rho_1,X_{ij},\gamma_{ij,n})\right|_{b=\beta}}{\Psi(\theta_0,X_{ij},\gamma_{ij,n})(1-\Psi(\theta_0,X_{ij},\gamma_{ij,n}))}
\\
&=&
\frac{1}{n^2}\sum_{i,j}\frac{\left.\frac{\partial}{\partial b}\Psi(b,\rho_0,\rho_1,X_{ij},\gamma_{ij,n})\right|_{b=\beta}\left.\frac{\partial}{\partial b'}\Psi(b,\rho_0,\rho_1,X_{ij},\gamma_{ij,n})\right|_{b=\beta}}{\Psi(\theta_0,X_{ij},\gamma_{ij,n})(1-\Psi(\theta_0,X_{ij},\gamma_{ij,n}))}\\
&=&
\frac{(1-\rho_0-\rho_1)^2}{n^2}\sum_{i,j}\frac{
\phi((Z_{ij,n}^\ast)'\beta_0)^2
}{\Psi(\theta_0,X_{ij},\gamma_{ij,n})(1-\Psi(\theta_0,X_{ij},\gamma_{ij,n}))}
Z_{ij,n}^\ast(Z_{ij,n}^\ast)'.
\end{eqnarray*}
Note that the minimum eigenvalue of $\sum_{i,j}Z_{ij,n}^\ast(Z_{ij,n}^\ast)'$ is bounded away from zero. 
Since $$\frac{\phi((Z_{ij,n}^\ast)'\beta_0)^2}{\Psi(\theta_0,X_{ij},\gamma_{ij,n})(1-\Psi(\theta_0,X_{ij},\gamma_{ij,n}))}$$ is bounded from zero uniformly over $i,j,n$, the minimum eigenvalue of 
$$
\sum_{i,j}\frac{
\phi((Z_{ij,n}^\ast)'\beta_0)^2
}{\Psi(\theta_0,X_{ij},\gamma_{ij,n})(1-\Psi(\theta_0,X_{ij},\gamma_{ij,n}))}
Z_{ij,n}^\ast(Z_{ij,n}^\ast)'
$$
is bounded away from zero. 
\end{proof}

\begin{lemma}\label{lemma:B_unif_conv}
$\sup_{b\in\mathcal{B}}\left\|E\left[\frac{\partial^2}{\partial b\partial b'}\boldsymbol{Q}_n(b)\mid X,\sigma_n\right]-\frac{\partial^2}{\partial b\partial b'}\hat{\boldsymbol{Q}}_n(b,\rho_0,\rho_1)\right\|=o_p(1)$ given $(X,\sigma_n)$.
\end{lemma}
\begin{proof}
By Lemma \ref{lemma:unfi_conv}, we need to show that
$$
\sup_{b\in\mathcal{B}}\left\|\frac{\partial^2}{\partial b\partial b'}\hat{\boldsymbol{Q}}_n(b,\rho_0,\rho_1)-\frac{\partial^2}{\partial b\partial b'}\boldsymbol{Q}_n(b)\right\|=o_p(1),
$$
that is, 
$$
\sup_{b\in\mathcal{B}}\left\|\boldsymbol{u}'\left(\frac{\partial^2}{\partial b\partial b'}\hat{\boldsymbol{Q}}_n(b,\rho_0,\rho_1)-\frac{\partial^2}{\partial b\partial b'}\boldsymbol{Q}_n(b)\right)\right\|=o_p(1)
\mbox{ for every vector }\boldsymbol{u}.
$$
Since 
\begin{eqnarray*}
&& 
\boldsymbol{u}'\left(\frac{\partial^2}{\partial b\partial b'}\hat{\boldsymbol{Q}}_n(b,\rho_0,\rho_1)-\frac{\partial^2}{\partial b\partial b'}\boldsymbol{Q}_n(b)\right)
\\
&=&
\frac{1}{n^2}\sum_{i,j}G_{ij,n}\boldsymbol{u}'\left(\frac{\partial}{\partial b'}(\boldsymbol{v}_1(b,\rho_0,\rho_1,X_{ij},\hat\gamma(X_{ij}))-\boldsymbol{v}_1(b,\rho_0,\rho_1,X_{ij},\gamma_{ij,n}))\right)
\\&&-
\frac{1}{n^2}\sum_{i,j}\boldsymbol{u}'\left(\frac{\partial}{\partial b'}(\boldsymbol{v}_2(b,\rho_0,\rho_1,X_{ij},\hat\gamma(X_{ij}))-\boldsymbol{v}_2(b,\rho_0,\rho_1,X_{ij},\gamma_{ij,n}))\right)
\\
&=&
\frac{1}{n^2}\sum_{i,j}G_{ij,n}\frac{\partial}{\partial b'}(\boldsymbol{u}'\boldsymbol{v}_1(b,\rho_0,\rho_1,X_{ij},\hat\gamma(X_{ij}))-\boldsymbol{u}'\boldsymbol{v}_1(b,\rho_0,\rho_1,X_{ij},\gamma_{ij,n}))
\\&&-
\frac{1}{n^2}\sum_{i,j}\frac{\partial}{\partial b'}(\boldsymbol{u}'\boldsymbol{v}_2(b,\rho_0,\rho_1,X_{ij},\hat\gamma(X_{ij}))-\boldsymbol{u}'\boldsymbol{v}_2(b,\rho_0,\rho_1,X_{ij},\gamma_{ij,n})),
\end{eqnarray*}
we have  
\begin{eqnarray*}
&&   
\left\|\boldsymbol{u}'\left(\frac{\partial^2}{\partial b\partial b'}\hat{\boldsymbol{Q}}_n(b,\rho_0,\rho_1)-\frac{\partial^2}{\partial b\partial b'}\boldsymbol{Q}_n(b)\right)\right\|
\\
&\leq&
\frac{1}{n^2}\sum_{i,j}\left\|\frac{\partial}{\partial b'}(\boldsymbol{u}'\boldsymbol{v}_1(b,\rho_0,\rho_1,X_{ij},\hat\gamma(X_{ij}))-\boldsymbol{u}'\boldsymbol{v}_1(b,\rho_0,\rho_1,X_{ij},\gamma_{ij,n}))\right\|
\\&&+
\frac{1}{n^2}\sum_{i,j}\left\|\frac{\partial}{\partial b'}(\boldsymbol{u}'\boldsymbol{v}_2(b,\rho_0,\rho_1,X_{ij},\hat\gamma(X_{ij}))-\boldsymbol{u}'\boldsymbol{v}_2(b,\rho_0,\rho_1,X_{ij},\gamma_{ij,n}))\right\|
\\
&\leq&
\frac{1}{n^2}\sum_{i,j}\sup_{\check{\gamma}_{ij}}\left\|\frac{\partial^2}{\partial b\partial \check{\gamma}_{ij}'}\boldsymbol{u}'\boldsymbol{v}_1(b,\rho_0,\rho_1,X_{ij},\check{\gamma}_{ij})\right\|\|\hat\gamma(X_{ij})-\gamma_{ij,n}\|
\\&&+
\frac{1}{n^2}\sum_{i,j}\sup_{\check{\gamma}_{ij}}\left\|\frac{\partial^2}{\partial b\partial \check{\gamma}_{ij}'}\boldsymbol{u}'\boldsymbol{v}_2(b,\rho_0,\rho_1,X_{ij},\check{\gamma}_{ij})\right\|\|\hat\gamma(X_{ij})-\gamma_{ij,n}\|
\\
&\leq&
\sup_{i,j}\sup_{\check{\gamma}_{ij}}\left\|\frac{\partial^2}{\partial b\partial \check{\gamma}_{ij}'}\boldsymbol{u}'\boldsymbol{v}_1(b,\rho_0,\rho_1,X_{ij},\check{\gamma}_{ij})\right\|\sup_{i,j}\|\hat\gamma(X_{ij})-\gamma_{ij,n}\|
\\&&+
\sup_{i,j}\sup_{\check{\gamma}_{ij}}\left\|\frac{\partial^2}{\partial b\partial \check{\gamma}_{ij}'}\boldsymbol{u}'\boldsymbol{v}_2(b,\rho_0,\rho_1,X_{ij},\check{\gamma}_{ij})\right\|\sup_{i,j}\|\hat\gamma(X_{ij})-\gamma_{ij,n}\|.
\end{eqnarray*}
Since $\frac{\partial^2}{\partial b\partial \check{\gamma}_{ij}'}\boldsymbol{u}'\boldsymbol{v}_1$ and $\frac{\partial^2}{\partial b\partial \check{\gamma}_{ij}'}\boldsymbol{u}'\boldsymbol{v}_2$ have bounded supports, we have  
$$
\sup_{b\in\mathcal{B}}\left\|\boldsymbol{u}'\left(\frac{\partial^2}{\partial b\partial b'}\hat{\boldsymbol{Q}}_n(b,\rho_0,\rho_1)-\frac{\partial^2}{\partial b\partial b'}\boldsymbol{Q}_n(b)\right)\right\|=O_p\left(\sup_{i,j}\|\hat\gamma(X_{ij})-\gamma_{ij,n}\|\right)=o_p(1).
$$
\end{proof}

\begin{lemma}\label{lemma:QMLE_asynormal}
\begin{eqnarray*}
   \sqrt{n}E\left[\frac{1}{n}\sum_{k=1}^n\psi_{\boldsymbol{Q},k, n}\psi_{\boldsymbol{Q},k, n}'\mid X,\sigma_n\right]^{-1/2}\left.E\left[\frac{\partial^2}{\partial b\partial b'}\boldsymbol{Q}_n(b)\mid X,\sigma_n\right]\right|_{b=\beta}(\hat\beta(\rho_0,\rho_1)-\beta)\rightarrow_{d}N(0,I)
   \mbox{ given }(X,\sigma_n).   
\end{eqnarray*}
\end{lemma}
\begin{proof} 
By Lemma \ref{lemma:QMLE_consisitent}, \ref{lemma:min_eig_hess}, \ref{lemma:B_unif_conv} and \citet[Theorem 5.1]{white1988unified}, it suffices to prove the following statements: 
\begin{itemize}
\item $\left.E\left[\frac{\partial^2}{\partial b\partial b'}\boldsymbol{Q}_n(b)\mid X,\sigma_n\right]\right|_{b=\beta}$ and $E\left[\frac{1}{n}\sum_{k=1}^n\psi_{\boldsymbol{Q},k, n}\psi_{\boldsymbol{Q},k, n}'\right]$ are $O(1)$; 
\item $E\left[\frac{\partial^2}{\partial b\partial b'}\boldsymbol{Q}_n(b)\mid X,\sigma_n\right]$ is continuous in $b\in\mathcal{B}$ uniformly in $n$; and 
\item 
\begin{equation}\label{eq:conv_norm_Qder}
\sqrt{n}E\left[\frac{1}{n}\sum_{k=1}^n\psi_{\boldsymbol{Q},k, n}\psi_{\boldsymbol{Q},k, n}'\mid X,\sigma_n\right]^{-1/2}\left.\frac{\partial}{\partial b}\hat{\boldsymbol{Q}}_n(b,\rho_0,\rho_1)\right|_{b=\beta}\rightarrow_dN(0,I)
\end{equation}
\end{itemize}
The first two statements follow from the error being normally distributed.

Before proving Eq. (\ref{eq:conv_norm_Qder}), we are going to show that 
\begin{equation}
\frac{1}{n^2}\sum_{i,j}(G_{ij,n}-E\left[G_{ij,n}\mid X,\sigma_n\right])\boldsymbol{V}_1(\theta_0,X_{ij},\gamma_{ij,n})(\hat{\gamma}_{ij}-\gamma_{ij,n})=o_p(n^{-1/2}).
\label{eq:conv_norm_Qder_condition3}
\end{equation}
Note that  
\begin{eqnarray*}
&&
\frac{1}{n^2}\sum_{i,j}(G_{ij,n}-E\left[G_{ij,n}\mid X,\sigma_n\right])\boldsymbol{V}_1(\theta_0,X_{ij},\gamma_{ij,n})(\hat{\gamma}_{ij}-\gamma_{ij,n})
\\
&=&
\frac{1}{n^2}\sum_{i,j}(G_{ij,n}-E\left[G_{ij,n}\mid X,\sigma_n\right])\boldsymbol{V}_1(\theta_0,X_{ij},\gamma_{ij,n})\frac{1}{n}\sum_{k} \psi_{\gamma,k, n}(X_{ij})\\
&=&
\frac{1}{n^2}\sum_{i,k}\mathrm{term}(i,k),
\end{eqnarray*}
where 
\begin{eqnarray*}
   \mathrm{term}(i,k) &=& \frac{1}{n}\sum_{j}(G_{ij,n}-E\left[G_{ij,n}\mid X,\sigma_n\right])\boldsymbol{V}_1(\theta_0,X_{ij},\gamma_{ij,n}) \psi_{\gamma,k, n}(X_{ij}).   
\end{eqnarray*}
We will demonstrate the $L^2$ convergence of $\frac{1}{n^2}\sum_{i,k}\mathrm{term}(i,k)$.
Regarding the expectation of $\frac{1}{n^2}\sum_{i,k}\mathrm{term}(i,k)$, we have 
\begin{eqnarray*}
&&
E\left[\frac{1}{n^2}\sum_{i,k}\mathrm{term}(i,k)\mid X,\sigma_n\right]
\\
&=&
\frac{1}{n^2}\sum_{i\ne k}\frac{1}{n}\sum_{j}E\left[(G_{ij,n}-E\left[G_{ij,n}\mid X,\sigma_n\right])\boldsymbol{V}_1(\theta_0,X_{ij},\gamma_{ij,n}) \psi_{\gamma,k, n}(X_{ij})\mid X,\sigma_n\right]
\\
&&+
\frac{1}{n^2}\sum_{i}\frac{1}{n}\sum_{j}E\left[(G_{ij,n}-E\left[G_{ij,n}\mid X,\sigma_n\right])\boldsymbol{V}_1(\theta_0,X_{ij},\gamma_{ij,n}) \psi_{\gamma,i, n}(X_{ij})\mid X,\sigma_n\right]
\\
&=&
\frac{1}{n^2}\sum_{i\ne k}\frac{1}{n}\sum_{j}E\left[(G_{ij,n}-E\left[G_{ij,n}\mid X,\sigma_n\right])\mid X,\sigma_n\right]\boldsymbol{V}_1(\theta_0,X_{ij},\gamma_{ij,n}) E\left[\psi_{\gamma,k, n}(X_{ij})\mid X,\sigma_n\right]
\\
&&+
\frac{1}{n^2}\sum_{i}\frac{1}{n}\sum_{j}E\left[(G_{ij,n}-E\left[G_{ij,n}\mid X,\sigma_n\right])\boldsymbol{V}_1(\theta_0,X_{ij},\gamma_{ij,n}) \psi_{\gamma,i, n}(X_{ij})\mid X,\sigma_n\right]
\\
&=&
\frac{1}{n^2}\sum_{i}\frac{1}{n}\sum_{j}E\left[(G_{ij,n}-E\left[G_{ij,n}\mid X,\sigma_n\right])\boldsymbol{V}_1(\theta_0,X_{ij},\gamma_{ij,n}) \psi_{\gamma,i, n}(X_{ij})\mid X,\sigma_n\right]
\\
&=&
O(n^{-1}).
\end{eqnarray*}
where the second equality comes from the independence of $\{G_{ij,n}:j\}$ across $i$, the third equality follows from $E\left[(G_{ij,n}-E\left[G_{ij,n}\mid X,\sigma_n\right])\mid X,\sigma_n\right]=0$, and the last equality follows because $G_{ij,n}$, $\boldsymbol{V}_1(\theta_0,X_{ij},\gamma_{ij,n})$, and  $\psi_{\gamma,i, n}(X_{ij})$ are bounded.  
Regarding the variance of $\frac{1}{n^2}\sum_{i,k}\mathrm{term}(i,k)$, we use  
\begin{eqnarray}
&&Cov\left(\mathrm{term}(i_1,k_1),\mathrm{term}(i_2,k_2)\mid X,\sigma_n\right)=0\mbox{ if }k_2\ne i_1,k_1,i_2\label{eq_U_shape1}\\
&&Cov\left(\mathrm{term}(i_1,k_1),\mathrm{term}(i_2,k_2)\mid X,\sigma_n\right)=0\mbox{ if }k_1\ne i_1,k_2,i_2\label{eq_U_shape2}\\
&&Cov\left(\mathrm{term}(i_1,k_1),\mathrm{term}(i_2,k_1)\mid X,\sigma_n\right)=0\mbox{ if }i_1\ne k_1,i_2,\label{eq_U_shape3}
\end{eqnarray}
where they result from the fact that $G_{ij,n}-E\left[G_{ij,n}\mid X,\sigma_n\right]$ and $\psi_{\gamma,k, n}(X_{ij})$ are mean-zero and independent across $(i,k)$.
We have 
\begin{eqnarray*}
Var\left(\frac{1}{n^2}\sum_{i,k}\mathrm{term}(i,k)\mid X,\sigma_n\right)
&=&
\frac{1}{n^4}\sum_{(i_1,k_1,i_2,k_2)}Cov\left(\mathrm{term}(i_1,k_1),\mathrm{term}(i_2,k_2)\mid X,\sigma_n\right)
\\
&=&
\frac{1}{n^4}\sum_{(i_1,k_1,i_2)}\sum_{k_2=i_1,k_1,i_2}Cov\left(\mathrm{term}(i_1,k_1),\mathrm{term}(i_2,k_2)\mid X,\sigma_n\right)
\\
&=&
\frac{1}{n^4}\sum_{(i_1,i_2)}\sum_{k_2=i_1,i_2}\sum_{k_1}Cov\left(\mathrm{term}(i_1,k_1),\mathrm{term}(i_2,k_2)\mid X,\sigma_n\right)
\\
&&+
\frac{1}{n^4}\sum_{(k_1,i_2)}\sum_{i_1}Cov\left(\mathrm{term}(i_1,k_1),\mathrm{term}(i_2,k_1)\mid X,\sigma_n\right)
\\
&=&
\frac{1}{n^4}\sum_{(i_1,i_2)}\sum_{k_2=i_1,i_2}\sum_{k_1=i_1,i_2}Cov\left(\mathrm{term}(i_1,k_1),\mathrm{term}(i_2,k_2)\mid X,\sigma_n\right)
\\
&&+
\frac{1}{n^4}\sum_{(k_1,i_2)}\sum_{i_1= k_1,i_2}Cov\left(\mathrm{term}(i_1,k_1),\mathrm{term}(i_2,k_1)\mid X,\sigma_n\right)
\\
&\leq&
\frac{6}{n^2}
\max_{(i_1,k_1,i_2,k_2)}|Cov\left(\mathrm{term}(i_1,k_1),\mathrm{term}(i_2,k_2)\mid X,\sigma_n\right)|
\\
&=&
O(n^{-2}),
\end{eqnarray*}
where the third equality uses Eq.(\ref{eq_U_shape1}), the fifth equality uses Eq.(\ref{eq_U_shape2}) and Eq.(\ref{eq_U_shape3}), and the last equality follows from $\sup_{(i_1,k_1,i_2,k_2)}\left|Cov\left(\mathrm{term}(i_1,k_1),\mathrm{term}(i_2,k_2)\mid X,\sigma_n\right)\right|=O(1)$.

Now we are going to show that Eq. (\ref{eq:conv_norm_Qder_condition3}) implies Eq. (\ref{eq:conv_norm_Qder}). 
The first-order Taylor expansions yield
\begin{eqnarray*}
\sup_{i,j}\|\boldsymbol{v}_1(\theta_0,X_{ij},\hat{\gamma}_{ij})-\boldsymbol{v}_1(\theta_0,X_{ij},\gamma_{ij,n})-\boldsymbol{V}_1(\theta_0,X_{ij},\gamma_{ij,n})(\hat{\gamma}_{ij}-\gamma_{ij,n})\|&=&o_p\left(\sup_{i,j}\|\hat{\gamma}_{ij}-\gamma_{ij,n}\|\right)
\\
\sup_{i,j}\|\boldsymbol{v}_2(\theta_0,X_{ij},\hat{\gamma}_{ij})-\boldsymbol{v}_2(\theta_0,X_{ij},\gamma_{ij,n})-\boldsymbol{V}_2(\theta_0,X_{ij},\gamma_{ij,n})(\hat{\gamma}_{ij}-\gamma_{ij,n})\|&=&o_p\left(\sup_{i,j}\|\hat{\gamma}_{ij}-\gamma_{ij,n}\|\right),
\end{eqnarray*}
and 
\begin{eqnarray*}
&&
\left.\frac{\partial}{\partial b}\hat{\boldsymbol{Q}}_n(b,\rho_0,\rho_1)\right|_{b=\beta}
\\
&&=
\frac{1}{n^2}\sum_{i,j}(G_{ij,n}-\Psi(\theta_0,X_{ij},\gamma_{ij,n}))\boldsymbol{v}_1(\theta_0,X_{ij},\gamma_{ij,n})
\\&&\quad+
\frac{1}{n^2}\sum_{i,j}\left(G_{ij,n}\left(
\boldsymbol{v}_1(\theta_0,X_{ij},\hat\gamma(X_{ij}))-\boldsymbol{v}_1(\theta_0,X_{ij},\gamma_{ij,n})
\right)
-\left(
\boldsymbol{v}_2(\theta_0,X_{ij},\hat\gamma(X_{ij}))-\boldsymbol{v}_2(\theta_0,X_{ij},\gamma_{ij,n})\right)
\right)\\
&&=
\frac{1}{n^2}\sum_{i,j}(G_{ij,n}-\Psi(\theta_0,X_{ij},\gamma_{ij,n}))
\boldsymbol{v}_1(\theta_0,X_{ij},\gamma_{ij,n})\\&&\quad+
\frac{1}{n^2}\sum_{i,j}\left(G_{ij,n}\boldsymbol{V}_1(\theta_0,X_{ij},\gamma_{ij,n})-\boldsymbol{V}_2(\theta_0,X_{ij},\gamma_{ij,n})\right)(\hat{\gamma}_{ij}-\gamma_{ij,n})+o_p(n^{-1/2})
\\
&&=
\frac{1}{n^2}\sum_{i,j}(G_{ij,n}-\Psi(\theta_0,X_{ij},\gamma_{ij,n}))
\boldsymbol{v}_1(\theta_0,X_{ij},\gamma_{ij,n})\\&&\quad+
\frac{1}{n^2}\sum_{i,j}\left(E\left[G_{ij,n}\mid X,\sigma_n\right]\boldsymbol{V}_1(\theta_0,X_{ij},\gamma_{ij,n})-\boldsymbol{V}_2(\theta_0,X_{ij},\gamma_{ij,n})\right)(\hat{\gamma}_{ij}-\gamma_{ij,n})+o_p(n^{-1/2})
\\
&&=
\frac{1}{n^2}\sum_{i,j}(G_{ij,n}-\Psi(\theta_0,X_{ij},\gamma_{ij,n}))
\boldsymbol{v}_1(\theta_0,X_{ij},\gamma_{ij,n})\\&&\quad+
\frac{1}{n^2}\sum_{i,j}\left(E\left[G_{ij,n}\mid X,\sigma_n\right]\boldsymbol{V}_1(\theta_0,X_{ij},\gamma_{ij,n})-\boldsymbol{V}_2(\theta_0,X_{ij},\gamma_{ij,n})\right)\frac{1}{n}\sum_{k} \psi_{\gamma,k, n}(X_{ij})+o_p(n^{-1/2})
\\
&&=
\frac{1}{n}\sum_{k}\psi_{\boldsymbol{Q},k, n}+o_p(n^{-1/2}),
\end{eqnarray*}
where the third equality uses Eq. (\ref{eq:conv_norm_Qder_condition3}).
We can apply Lyapunov’s central limit theorem  to uniformly-bounded random variables $\psi_{\boldsymbol{Q},k, n}$, and we have Eq. (\ref{eq:conv_norm_Qder}).
\end{proof}

\begin{lemma}\label{lemma:clt_hatbeta}
$\sqrt{n}\widehat{AV}(\rho_0,\rho_1)^{-1/2}(\hat\beta(\rho_0,\rho_1)-\beta)\rightarrow_{d}N(0,I)\mbox{ given }(X,\sigma_n)$.
\end{lemma}
\begin{proof}
By Lemma \ref{lemma:QMLE_asynormal}, it is sufficient to show that 
\begin{eqnarray*}
\left.\frac{\partial^2}{\partial b\partial b'}\hat{\boldsymbol{Q}}_n(b,\rho_0,\rho_1)\right|_{b=\hat\beta(\rho_0,\rho_1)}-\left.E\left[\frac{\partial^2}{\partial b\partial b'}\boldsymbol{Q}_n(b)\mid X,\sigma_n\right]\right|_{b=\beta}&=&o_p(1)\\
\mathcal{\hat{S}}(\rho_0,\rho_1)-E\left[\frac{1}{n}\sum_{k=1}^n\psi_{\boldsymbol{Q},k, n}\psi_{\boldsymbol{Q},k, n}'\mid X,\sigma_n\right]&=&o_p(1).
\end{eqnarray*}
The first statement follows from Lemma \ref{lemma:B_unif_conv}.
The second statement can be shown similarly to Lemma \ref{var_conv}.
\end{proof}

\section{Covariate-Dependent Misclassification}\label{sec_aprob}

In this section, we weaken Assumption \ref{Ass3} into the following assumption.

\begin{assumption}
   \label{Ass3_Xdependent}
   The following two statements hold for every $n$ and every $i,j,k\in\mathcal{N}_n$.
   (i) $G_{ki,n}$ and $G_{kj,n}$ are independent given $(G_{ki,n}^\ast, G_{kj,n}^\ast,X,\sigma_n)$. 
   (ii) $Pr(G_{ij,n}\ne G_{ij,n}^\ast\mid G_{ij,n}^\ast,X,\sigma_n)=\rho_0(X_i)1\{G_{ij,n}^\ast=0\}+\rho_1(X_i)1\{G_{ij,n}^\ast=1\}$, 
   with $\rho_0(X_i), \rho_1(X_i) \geq 0$ and $\rho_0(X_i) + \rho_1(X_i) < 1$ a.s.
\end{assumption}

Assumption \ref{Ass3_Xdependent} implies that the probability of misclassification can be covariate-specific.  This setting captures heterogeneity in the misclassification due to differences in the observed types and allows for correlation within the misclassification of links for a given individual, i.e., the misclassification process of $G_{ij}$ and $G_{ik}$ could be correlated through $X_i$. This specification might be desirable if the researcher is concerned that individuals of certain profiles (e.g., age, religion, caste) are more prone to fatigue
or apprehensive about listing all their individual connections.

Under Assumption \ref{Ass3_Xdependent}, we can modify Lemma \ref{belief_lemma} as follows. 

\begin{lemma}
If Assumptions \ref{Ass1}, \ref{Ass2} and \ref{Ass3_Xdependent} hold, then for every distinct $i,j,k \in \mathcal{N}_n$, 
\begin{eqnarray*}
   \gamma_{ij,n}^\ast
   &=&
   \left(
   \begin{array}{c}
   E\left[G_{ji,n}^\ast \mid X, \sigma_n\right]\\
   \frac{1}{n}\sum_{k}E\left[G_{kj,n}^\ast\mid X, \sigma_n\right]\\
   \frac{1}{n}\sum_{k}E\left[G_{ki,n}^\ast G_{kj,n}^\ast \mid X, \sigma_n\right]
   \end{array}
   \right)
   \\
   &=&
   \left(
   \begin{array}{c}
   \omega(X_j)^{-1} 
   \left\{  E\left[ G_{ji,n} \mid X, \sigma_n \right] - \rho_0(X_{j})  \right\}
   \\
   \frac{1}{n}\sum_{k}
   \omega(X_k)^{-1} 
   \left\{ E\left[ G_{ki,n} \mid X, \sigma_n \right] - \rho_0(X_{k})  \right\}
   \\
   \frac{1}{n}\sum_{k}
   \omega(X_k)^{-2} 
   \left\{  
   E\left[ G_{ki,n} G_{kj,n} \mid X, \sigma_n \right]
   -
   \rho_0(X_{k}) 
   E\left[ G_{ki,n}+G_{kj,n} \mid X, \sigma_n \right]
   + 
   \rho_0(X_{k})^2
   \right\}
   \end{array}
   \right)
\end{eqnarray*}
where $\omega(X_k) = (1-\rho_0(X_{k})-\rho_1(X_{k}))$.

\end{lemma}
\begin{proof}
Notice that 
\begin{eqnarray*}
   E\left[ G_{ki,n}^\ast \mid X, \sigma_n \right]   
   &=&
   ( 1- \rho_0(X_{k}) - \rho_1(X_{k}) )^{-1}
   \left\{  
   E\left[ G_{ki,n} \mid X, \sigma_n \right]
   -
   \rho_0(X_{k}) 
   \right\}
\end{eqnarray*}
and 
\begin{eqnarray*}
   E\left[ G_{ki,n}^\ast + G_{kj,n}^\ast \mid X, \sigma_n \right]   
   &=&
   ( 1- \rho_0(X_{k}) - \rho_1(X_{k}) )^{-1}
   \left\{  
   E\left[ G_{ki,n}+G_{kj,n} \mid X, \sigma_n \right]
   -
   2\rho_0(X_{k}) 
   \right\}.
\end{eqnarray*}
Consider 
\begin{eqnarray*}
   &&
   E\left[ G_{ki,n} G_{kj,n} \mid X, \sigma_n \right]
   \\
   &=&
   E
   \left[  
      Pr\left[ G_{ki,n}=1, G_{kj,n} = 1 \mid X, \sigma_n, G_{ki,n}^\ast G_{kj,n}^\ast \right]
      \mid 
      X, \sigma_n
   \right]
   \\
   &=&
   E
   \left[  
      Pr\left[ G_{ki,n}  = 1 \mid X, \sigma_n, G_{ki,n}^\ast G_{kj,n}^\ast \right]
      Pr\left[ G_{kj,n} = 1 \mid X, \sigma_n, G_{ki,n}^\ast G_{kj,n}^\ast \right]
      \mid 
      X, \sigma_n
   \right]
   \\
   &=&
   ( 1-\rho_1(X_{k}) )^2
   Pr\left[ G_{ki,n}^\ast = 1,  G_{kj,n}^\ast =1 \mid X, \sigma_n \right]
   \\
   &&
   +
   \rho_0(X_{k}) ( 1-\rho_1(X_{k}) ) 
   \left\{   
      Pr\left[ G_{ki,n}^\ast= 1,  G_{kj,n}^\ast=0  \mid X, \sigma_n \right]
      +
      Pr\left[ G_{ki,n}^\ast= 0,  G_{kj,n}^\ast=1  \mid X, \sigma_n \right]
   \right\}
   \\
   &&
   +
   \rho_0(X_{k})^2
   Pr\left[ G_{ki,n}^\ast= 0,  G_{kj,n}^\ast=0  \mid X, \sigma_n \right]
   \\
   &=&
   \rho_0(X_{k})^2
   +
   \left\{  
   \rho_0(X_{k}) ( 1-\rho_0(X_{k})-\rho_1(X_{k}) ) 
   \right\}
   E\left[ G_{ki,n}^\ast+ G_{kj,n}^\ast \mid X, \sigma_n\right]
   \\
   &&
   +
   ( 1-\rho_0(X_{k})-\rho_1(X_{k}) )^2 
   E\left[ G_{ki,n}^\ast G_{kj,n}^\ast \mid X, \sigma_n\right], 
\end{eqnarray*}   
where the fourth equality follows from using 
\begin{eqnarray*}
   Pr\left[ G_{ki,n}^\ast = 1,  G_{kj,n}^\ast =1 \mid X, \sigma_n \right] 
   &=& E\left[ G_{ki,n}^\ast G_{kj,n}^\ast \mid X, \sigma_n\right] \\
   Pr\left[ G_{ki,n}^\ast= 1,  G_{kj,n}^\ast=0  \mid X, \sigma_n \right] 
   &=& E\left[ G_{ki,n}^\ast \mid X, \sigma_n\right] - E\left[ G_{ki,n}^\ast G_{kj,n}^\ast \mid X, \sigma_n\right]\\
   Pr\left[ G_{ki,n}^\ast= 0,  G_{kj,n}^\ast=1  \mid X, \sigma_n \right] 
   &=& E\left[ G_{kj,n}^\ast \mid X, \sigma_n\right] - E\left[ G_{ki,n}^\ast G_{kj,n}^\ast \mid X, \sigma_n\right]\\
   Pr\left[ G_{ki,n}^\ast= 0,  G_{kj,n}^\ast=0  \mid X, \sigma_n \right]
   &=&
   1 - E\left[ G_{ki,n}^\ast + G_{kj,n}^\ast \mid X, \sigma_n\right] + E\left[ G_{ki,n}^\ast G_{kj,n}^\ast \mid X, \sigma_n\right]
\end{eqnarray*}
Solving for $E\left[ G_{ki,n}^\ast G_{kj,n}^\ast \mid X, \sigma_n\right]$ we obtain 
\begin{eqnarray*}
   E\left[ G_{ki,n}^\ast G_{kj,n}^\ast \mid X, \sigma_n\right]
   &=&
   ( 1-\rho_0(X_{k})-\rho_1(X_{k}) )^{-2} 
   \left\{  
   E\left[ G_{ki,n} G_{kj,n} \mid X, \sigma_n \right]
   - 
   \rho_0(X_{k})^2
   \right\}
   \\
   &&
   -
   \rho_0(X_{k}) ( 1-\rho_0(X_{k})-\rho_1(X_{k}) )^{-1} 
   E\left[ G_{ki,n}^\ast+ G_{kj,n}^\ast \mid X, \sigma_n\right]
   \\
   &=&
   ( 1-\rho_0(X_{k})-\rho_1(X_{k}) )^{-2} 
   \left\{  
      \left(  
      E\left[ G_{ki,n} G_{kj,n} \mid X, \sigma_n \right]
      +
      \rho_0(X_{k})^2
      \right)  
   \right. 
   \\
   &&
   -
   \left.
   \rho_0(X_{k}) 
   E\left[ G_{ki,n}+G_{kj,n} \mid X, \sigma_n \right]
   \right\}.
\end{eqnarray*} 
\end{proof}

\begin{remark}
We can further weaken Condition (ii) in Assumption \ref{Ass3_Xdependent} to $Pr(G_{ij,n}\ne G_{ij,n}^\ast\mid G_{ij,n}^\ast,X,\sigma_n)=\rho_0(X)1\{G_{ij,n}^\ast=0\}+\rho_1(X)1\{G_{ij,n}^\ast=1\}$, 
   with $\rho_0(X), \rho_1(X) \geq 0$ and $\rho_0(X) + \rho_1(X) < 1$ a.s. In this case, the network statistics can be written as 
   \begin{eqnarray*}
      \gamma_{ij,n}^\ast
      &=&
      \left(
      \begin{array}{c}
      E\left[G_{ji,n}^\ast \mid X, \sigma_n\right]\\
      \frac{1}{n}\sum_{k}E\left[G_{kj,n}^\ast\mid X, \sigma_n\right]\\
      \frac{1}{n}\sum_{k}E\left[G_{ki,n}^\ast G_{kj,n}^\ast \mid X, \sigma_n\right]
      \end{array}
      \right)
      \\
      &=&
      D(p_0(X), p_1(X))
      \left(
      \begin{array}{c}
      E\left[ G_{ji,n} \mid X, \sigma_n \right] 
      \\
      \frac{1}{n}\sum_{k}
      E\left[ G_{ki,n} \mid X, \sigma_n \right] 
      \\
      \omega(X)^{-2} 
      \frac{1}{n}\sum_{k}
      E\left[ G_{ki,n} G_{kj,n} \mid X, \sigma_n \right]
      \\
      \omega(X)^{-2} 
      \rho_0(X) 
      \frac{1}{n}\sum_{k}
      E\left[ G_{ki,n}+G_{kj,n} \mid X, \sigma_n \right]
      \end{array}
      \right)
      +
      \left(
         \begin{array}{c}
         -\omega(X)^{-1} 
         \rho_0(X) 
         \\
         -\omega(X)^{-1} 
         \rho_0(X)  
         \\
         \omega(X)^{-2} 
         \rho_0(X)^2
         \end{array}
         \right),
   \end{eqnarray*}
where $\omega(X) = (1-\rho_0(X)-\rho_1(X))$ and 
\begin{eqnarray*}
   D(p_0(X), p_1(X)) &=&
   \begin{bmatrix}
      \omega(X)^{-1} & 0 &0 & 0 \\
      0 & \omega(X)^{-1} &0 & 0 \\
      0 & 0 &\omega(X)^{-2} & -\omega(X)^{-2}\rho_0(X)  \\
   \end{bmatrix}.
\end{eqnarray*}
\end{remark}

\section{Semiparametric Identification Analysis}\label{sec_a2}

Given $P\in\mathcal{P}$, we will characterize the identified set in the semiparametric model. 
\begin{definition}
For each distribution $P\in\mathcal{P}$, 
the identified set $\Theta_{I,SP}(P)$ is defined as the set of all $\theta=(b,r_0,r_1)$ in $\Theta$ for which there is some joint distribution $P^\ast\in\mathcal{P}^\ast$ such that 
Condition \ref{independence_assn}, \ref{linear_index}(ii), and \ref{misclas_prop} holds, and that the distribution of $(G_{ij,n},X_{ij},\gamma_{ij,n})$ induced from $P^\ast$ is equal to $P$.
\end{definition}

\begin{theorem}
Given $P\in\mathcal{P}$, 
$\Theta_{I,SP}(P)$ is equal to the set of $\theta\in\Theta$ satisfying the following statements a.s. for some $r_0, r_1 \geq 0$ such that $r_0 + r_1 < 1$ and some weakly increasing and right-continuous function $\Lambda$:  
\begin{eqnarray}
&&r_0\leq E_P\left[G_{ij,n}\mid Z_{ij,n}\right]\mbox{ and }r_1\leq E_P\left[1-G_{ij,n}\mid Z_{ij,n}\right]\label{a_0_ineq}\\
&&E_P\left[G_{ij,n}\mid Z_{ij,n}\right]=\Lambda\left((c(r_0,r_1)+\gamma_{ij,n}'C(r_0,r_1))'b_1+X_{ij}'b_2\right).\label{SSSinequality}
\end{eqnarray}
\end{theorem}

\begin{proof}
First, we are going to show that every element $\theta$ of $\Theta_{I,SP}(P)$ satisfies the conditions in (\ref{a_0_ineq})-(\ref{SSSinequality}).
Let $(r_0, r_1)$ denote the misclassification probabilities.  
Denote by $\Lambda^\ast$ the cdf of $-\varepsilon_{ij}$ and define $\Lambda(v)=r_0+(1-r_0-r_1)\Lambda^\ast(v)$. 
By Lemma \ref{belief_lemma} and \ref{Ass3_lemma}, 
\begin{eqnarray*}
E_{P^\ast}\left[G_{ij,n}\mid Z_{ij,n}\right]
&=&
r_0+(1-r_0-r_1)E_{P^\ast}\left[G_{ij,n}^\ast\mid Z_{ij,n}\right]\\
&=&
r_0+(1-r_0-r_1)\Lambda^\ast((c(r_0,r_1)+\gamma_{ij,n}'C(r_0,r_1))'b_1+X_{ij}'b_2),
\end{eqnarray*}
and we have the condition (\ref{SSSinequality}).
Note that $\Lambda$ is weakly increasing and right-continuous, because $\Lambda^\ast$ is weakly increasing and right-continuous.
The two inequalities in (\ref{a_0_ineq}) are shown as follows: 
\begin{eqnarray*}
E_{P^\ast}\left[G_{ij,n}\mid Z_{ij,n}\right]
&=&
r_0+(1-r_0-r_1)E_{P^\ast}\left[G_{ij,n}^\ast\mid Z_{ij,n}\right]\geq r_0\\
E_{P^\ast}\left[1-G_{ij,n}\mid Z_{ij,n}\right]
&=&
r_1+(1-r_0-r_1)E_{P^\ast}\left[1-G_{ij,n}^\ast\mid Z_{ij,n}\right]\geq r_1,
\end{eqnarray*}
where the inequalities follow from $1-r_0-r_1\geq 0$.

Now, the rest of the proof is going to show that every element $\theta\in\Theta$ satisfying (\ref{a_0_ineq})-(\ref{SSSinequality}), belongs to $\Theta_{I,SP}(P)$. 
By the condition (\ref{SSSinequality}) as well as Condition (\ref{a_0_ineq}), there is a weakly increasing and right-continuous function $\Lambda:\mathbb{R}\rightarrow [r_0,1-r_1]$ such that 
\begin{equation}\label{equE}
E_P\left[G_{ij,n}\mid Z_{ij,n}\right]=\Lambda\left((c(r_0,r_1)+\gamma_{ij,n}'C(r_0,r_1))'b_1+X_{ij}'b_2\right).
\end{equation}
Denote by $\Lambda^\ast$ the cdf satisfying $\Lambda(v)=r_0+(1-r_0-r_1)\Lambda^\ast(v)$.

Define the joint distribution $P^\ast$ in the following way. 
Define the cdf of $\varepsilon_{ij}$ such that $\Lambda^\ast$ is the cdf of $-\varepsilon_{ij}$.
The conditional distribution of $(\gamma_{ij,n},\gamma_{ij,n}^\ast,X_{ij})$ given $\varepsilon_{ij}$ is 
\begin{equation}\label{equA}
P^\ast((\gamma_{ij,n},\gamma_{ij,n}^\ast,X_{ij})\in B \mid \varepsilon_{ij})=P((\gamma_{ij,n},c(r_0,r_1)+C(r_0,r_1)\gamma_{ij,n},X_{ij})\in B)
\end{equation}
for all the measurable sets $B$. 
The conditional distribution of $G_{ij,n}^\ast$ given $(\gamma_{ij,n},\gamma_{ij,n}^\ast,X_{ij},\varepsilon_{ij})$ is 
\begin{equation}\label{equB}
P^\ast(G_{ij,n}^\ast=1\mid \gamma_{ij,n},\gamma_{ij,n}^\ast,X_{ij},\varepsilon_{ij})=1\{(Z_{ij,n}^\ast)'b + \varepsilon_{ij}\geq 0\}.
\end{equation}
The conditional distribution of $G_{ij,n}$ given $(G_{ij,n}^\ast,\gamma_{ij,n},\gamma_{ij,n}^\ast,X_{ij},\varepsilon_{ij})$  is 
\begin{equation}\label{equC}
P^\ast(G_{ij,n}=1\mid G_{ij,n}^\ast,\gamma_{ij,n},\gamma_{ij,n}^\ast,X_{ij},\varepsilon_{ij})
=
\begin{cases}
r_0&\mbox{ if }G_{ij,n}^\ast=0\\
1-r_1&\mbox{ if }G_{ij,n}^\ast=1.
\end{cases}
\end{equation}

Note that $(P^\ast,\theta)$ satisfies Conditions 1(ii), 2 and 3, because Condition 1(ii) follows from Eq. (\ref{equA}), Condition 2 follows from Eq. (\ref{equB}),  Condition 3(i) follows from Eq. (\ref{equB}) and (\ref{equC}), and Condition 3(ii) follows from Eq. (\ref{equA}).

To conclude this proof, we are going to show that the distribution of $(G_{ij,n},X_{ij},\gamma_{ij,n})$ induced from $P^\ast$ is equal to $P$.
The distribution of $(X_{ij},\gamma_{ij,n})$ induced from $P^\ast$ is equal to that from $P$, by Eq. (\ref{equA}).
The equality of $P^\ast(G_{ij,n}=1\mid Z_{ij,n})=P(G_{ij,n}=1\mid Z_{ij,n})$ a.s. under $P^\ast$ is shown as follows. 
Note that  
\begin{equation}\label{equD}
\gamma_{ij,n}^\ast=c(r_0,r_1)+C(r_0,r_1)\gamma_{ij,n} \mbox{ a.s.  under }P^\ast
\end{equation}
Then 
\begin{eqnarray*}
P^\ast(G_{ij,n}=1\mid Z_{ij,n})
&=&
P^\ast(G_{ij,n}=1\mid Z_{ij,n},\gamma_{ij,n}^\ast)\\
&=&
r_0P^\ast(G_{ij,n}^\ast=0\mid Z_{ij,n},\gamma_{ij,n}^\ast)+(1-r_1)P^\ast(G_{ij,n}^\ast=1\mid Z_{ij,n},\gamma_{ij,n}^\ast)\\
&=&
r_0+(1-r_0-r_1)P^\ast(G_{ij,n}^\ast=1\mid Z_{ij,n},\gamma_{ij,n}^\ast)\\
&=&
r_0+(1-r_0-r_1)E_{P^\ast}[P^\ast(G_{ij,n}^\ast=1\mid Z_{ij,n},\gamma_{ij,n}^\ast,\varepsilon_{ij})\mid Z_{ij,n},\gamma_{ij,n}^\ast]\\
&=&
r_0+(1-r_0-r_1)P^\ast((Z_{ij,n}^\ast)'b + \varepsilon_{ij}\geq 0\mid Z_{ij,n},\gamma_{ij,n}^\ast)\\
&=&
r_0+(1-r_0-r_1)\Lambda^\ast((Z_{ij,n}^\ast)'b)\\
&=&
r_0+(1-r_0-r_1)\Lambda^\ast((c(r_0,r_1)+C(r_0,r_1)\gamma_{ij,n})'b_1+X_{ij}'b_2)\\
&=&
P(G_{ij,n}=1\mid Z_{ij,n}), 
\end{eqnarray*}
where the first and seventh equalities follow from Eq. (\ref{equD}), the second follows from Eq. (\ref{equC}), the fifth follows from Eq. (\ref{equB}), and the last follows from Eq. (\ref{equE}).
\end{proof}

\clearpage
\section{Additional Tables}\label{S:additionaltables}

\begin{table}[ht] 
\centering
\caption{95\% confidence intervals $\hat{\mathcal{C}}_n(\alpha)$ with  $r_0=0$.}
   \label{table:Ext_unionR1}
   \begin{tabular}{
    l
    >{{[}} 
    S[table-format = -1.3,table-space-text-pre={[}]
    @{,\,} 
    S[table-format = -2.3,table-space-text-post={]}]
    <{{]}} 
    >{{[}} 
    S[table-format = -1.3,table-space-text-pre={[}]
    @{,\,} 
    S[table-format = -2.3,table-space-text-post={]}]
    <{{]}} 
    >{{[}} 
    S[table-format = -1.3,table-space-text-pre={[}]
    @{,\,} 
    S[table-format = -2.3,table-space-text-post={]}]
    <{{]}} 
    >{{[}} 
    S[table-format = -1.3,table-space-text-pre={[}]
    @{,\,} 
    S[table-format = -2.3,table-space-text-post={]}]
    <{{]}} 
}
\hline
&\multicolumn{2}{c}{$r_1\leq 0.6$}&\multicolumn{2}{c}{$r_1\leq0.7$}&\multicolumn{2}{c}{$r_1\leq 0.8$} & \multicolumn{2}{c}{$r_1\leq 0.9$} \cr
&  \multicolumn{2}{c}{(1)}& \multicolumn{2}{c}{(2)}&\multicolumn{2}{c}{(3)}& \multicolumn{2}{c}{(4)}  \cr \hline\hline
Reciprocation     &0.823  & 1.676    &0.701  & 1.676   &0.694  & 1.676   &0.500  & 1.676   \cr
In degree         &12.292 & 33.099   &9.830  & 33.099  &6.934  & 33.099  &4.192  & 33.099  \cr
Supported trust   &12.939 & 110.258  &7.808  & 110.258 &3.400  & 110.258 &1.825  & 110.258 \cr
Constant          &-3.896 & -3.246   &-3.896 & -3.204  &-3.896 & -3.027  &-3.896 & -2.826  \cr
Same religion     &0.348  & 0.548    &0.348  & 0.591   &0.348  & 0.643   &0.348  & 0.730   \cr
Same sex          &0.565  & 0.793    &0.565  & 0.834   &0.565  & 0.851   &0.565  & 0.944   \cr
Same caste        &0.195  & 0.353    &0.195  & 0.373   &0.195  & 0.396   &0.195  & 0.449   \cr
Same language     &-0.022 & 0.079    &-0.029 & 0.079   &-0.044 & 0.079   &-0.065 & 0.079   \cr
Same family       &1.308  & 2.139    &1.308  & 2.381   &1.308  & 2.845   &1.308  & 6.118   \cr \hline
\end{tabular}

   \caption*{\footnotesize{Note: $\hat{\mathcal{C}}_n(\alpha)$ is computed in Columns (1)-(4) as $\cup_{r_1\leq \bar{\mathcal{R}}_1} \mathcal{C}_n(\alpha; 0, r_1)$, with $\bar{\mathcal{R}}_1\in \{0.6, 0.7, 0.8, 0.9\}$.}}
\end{table}

\begin{table}[ht]       
\centering
\caption{Ratio of lengths of 95\% confidence intervals, $|\hat{\mathcal{C}}_n(\alpha)|/|\mathcal{C}_n(\alpha, 0,0)|$.}
   \label{table:Ext_ratiosR1}
   \begin{tabular}{l c c c c}
\hline
& $r_1\leq 0.6$ &  $r_1\leq 0.7$   &  $r_1\leq 0.8$ &  $r_1\leq 0.9$   \\ 
& (1)           &    (2)           &    (3)         &  (4)             \\ \hline\hline 
Reciprocation  &    2.555&2.919&2.941&3.522\\
In degree      &    3.088&3.453&3.883&4.290\\
Supported trust&    1.917&2.018&2.104&2.136\\
Constant       &    1.802&1.918&2.412&2.968\\
Same religion  &    1.386&1.687&2.050&2.652\\
Same sex       &    1.628&1.920&2.042&2.702\\
Same caste     &    1.387&1.558&1.765&2.230\\
Same language  &    1.190&1.263&1.445&1.685\\
Same family    &    3.628&4.682&6.708&20.987\\\hline
\end{tabular}

\end{table}

\clearpage
\bibliography{bib_missperp}
\end{document}